\documentclass[11pt]{article}

\usepackage[T1]{fontenc}
\usepackage[utf8]{inputenc}

\usepackage{lmodern}
\usepackage{xspace}
\usepackage{fullpage}
\usepackage{amsfonts,amsmath,amssymb,amsthm}
\usepackage[colorlinks=true, citecolor=blue]{hyperref}

\usepackage[capitalize,nameinlink]{cleveref}

\usepackage{mleftright}

\usepackage[dvipsnames]{xcolor}

\usepackage{multirow}

\usepackage{dsfont} 
\usepackage{fullpage}

\providecommand{\email}[1]{\href{mailto:#1}{\nolinkurl{#1}\xspace}}

\usepackage{algorithmicx,algpseudocode,  algorithm}

\usepackage[shortlabels]{enumitem}
\setitemize{noitemsep,topsep=0pt,parsep=0pt,partopsep=0pt}

\makeatletter
\@ifundefined{theorem}{
  \theoremstyle{definition}
  \newtheorem{definition}{Definition}
  \theoremstyle{plain}
  \newtheorem{result}{Result}
  \newtheorem{theorem}{Theorem}
  \newtheorem{corollary}[theorem]{Corollary}
  \newtheorem{lemma}[theorem]{Lemma}
  \newtheorem{claim}[theorem]{Claim}
  \newtheorem{fact}[theorem]{Fact}
  
  \theoremstyle{remark}
  \newtheorem{remark}{Remark}

}{}
\makeatother
\newenvironment{proofof}[1]{\begin{proof}[Proof of {#1}]}{\end{proof}}

\newcommand{\EE}{\mathbb{E}}

\newcommand{\RR}{\mathbb{R}}

\newcommand{\setOfSuchThat}[2]{ \left\{\; #1 \;\colon\; #2\; \right\} } 			
\newcommand{\indicSet}[1]{\mathds{1}_{#1}}                                              
\newcommand{\indic}[1]{\indicSet{\left\{#1\right\}}}                                             

\newcommand{\expect}[1]{\EE\mleft[#1\mright]}

\newcommand{\expectCond}[2]{\mathbb{E}\mleft[\, #1 \;\middle\vert\; #2\, \mright]}
\newcommand{\shortexpect}{\EE}
\newcommand{\var}{\operatorname{Var}}
\newcommand{\cov}{\operatorname{Cov}}
\newcommand{\variance}[1]{\var\mleft(#1\mright)}
\newcommand{\covariance}[2]{\cov\mleft(#1,#2\mright)}

\newcommand{\proba}{\Pr}
\newcommand{\probaOf}[1]{\proba\mleft[\, #1\, \mright]}
\newcommand{\probaCond}[2]{\proba\mleft[\, #1 \;\middle\vert\; #2\, \mright]}

\newcommand{\distribs}[1]{\Delta\mleft(#1\mright)} 

\def \by     {{\bf y}}
\def \bb     {{\bf b}}

\def \cT     {{\cal T}}

\def \cX     {{\cal X}}
\def \cY     {{\cal Y}}
\def \cZ     {{\cal Z}}

\def \ceil#1{{\lceil{#1}\rceil}}
\def \Ceil#1{{\left\lceil{#1}\right\rceil}}

\def \Paren#1{\left(#1\right)}

\newcommand{\eqdef}{:=}

\newcommand{\norm}[1]{\lVert#1{\rVert}}
\newcommand{\normone}[1]{{\norm{#1}}_1}
\newcommand{\normtwo}[1]{{\norm{#1}}_2}

\newcommand{\normfour}[1]{{\norm{#1}}_4}
\newcommand{\norminf}[1]{{\norm{#1}}_\infty}
\newcommand{\abs}[1]{\left\lvert #1 \right\rvert}

\newcommand{\dotprod}[2]{ \mleft\langle #1,\xspace #2 \mright\rangle } 			

\newcommand{\yes}{{\sf{}yes}\xspace}
\newcommand{\no}{{\sf{}no}\xspace}
\newcommand{\matH}{\mathbf{H}}

\newcommand{\ns}{n} 
\newcommand{\ab}{k} 
\newcommand{\numbits}{\ell}

\newcommand{\mech}{W}
\newcommand{\eps}{\varepsilon}
\newcommand{\dst}{\gamma}

\newcommand{\p}{p}
\newcommand{\q}{q}

\newcommand{\lp}[1][1]{\ell_{#1}}

\newcommand{\uniform}{\ensuremath{u}}
\newcommand{\bernoulli}[1]{\ensuremath{\operatorname{Bern}\mleft( #1 \mright)}}
\newcommand{\bin}[2]{\ensuremath{\operatorname{Bin}\mleft( #1, #2 \mright)}}

\newcommand{\dtv}{\operatorname{d}_{\rm TV}}
\newcommand{\totalvardistrestr}[3][]{{\dtv^{#1}\!\left({#2, #3}\right)}}
\newcommand{\totalvardist}[2]{\totalvardistrestr[]{#1}{#2}}

\newcommand{\Rappor}{\textsc{Rappor}\xspace}
\newcommand{\Raptor}{\textsc{Raptor}\xspace}
\newcommand{\HR}{\textsc{HR}\xspace}
\newcommand{\RaR}{\textsc{RR}\xspace}

\title{Test without Trust: Optimal Locally Private Distribution Testing}
\author{
    Jayadev Acharya\thanks{Cornell University. Email: \email{acharya@cornell.edu}.}
      \and Cl\'ement L. Canonne\thanks{Stanford University. Email: \email{ccanonne@cs.stanford.edu}. Supported by a Motwani Postdoctoral Fellowship.}
      \and Cody Freitag\thanks{Cornell University. Email: \email{cfreitag@cs.cornell.edu}.}
      \and Himanshu Tyagi\thanks{Indian Institute of Science. Email: \email{htyagi@iisc.ac.in}.}
}

\begin{document}
\maketitle

\begin{abstract}
We study the problem of distribution testing when the samples
can only be accessed using a locally
differentially private mechanism and focus on two representative
testing questions of identity (goodness-of-fit) and independence
testing for discrete distributions. We are concerned with two
settings: First, when we insist on using an already deployed,  general-purpose
locally differentially private mechanism such as the popular \Rappor
or the recently introduced Hadamard Response for collecting 
data, and must build our tests based on the data collected via this mechanism; and second, when no 
such restriction is imposed, and we can design a bespoke mechanism 
specifically for testing. For the latter purpose, we introduce the
\emph{Randomized Aggregated Private Testing Optimal Response} (\Raptor) 
mechanism which  is remarkably simple 
and requires only one bit of communication per sample.

We propose tests based on these mechanisms and
analyze their sample complexities. Each proposed test  can be
implemented efficiently. In each case (barring one), we
complement our performance bounds for 
algorithms with 
information-theoretic lower bounds and establish sample optimality of
our proposed algorithm. A peculiar feature that emerges is that our
sample-optimal algorithm based on \Raptor uses public-coins, and any
test based on \Rappor or Hadamard Response, which are both
private-coin mechanisms, requires significantly more samples. 
\end{abstract}
\newpage

\section{Introduction}\label{sec:introduction}

Locally differentially private (LDP) mechanisms have gained prominence
as methods of choice for sharing sensitive data with untrusted
curators. This strong notion of privacy, introduced
in~\cite{DuchiJW13} (see also~\cite{EvfimievskiGS03}) as a variant
of differential privacy~\cite{DworkMNS06,Dwork06}, requires each user
to report only a noisy version of its data such that the distribution
of the reported data does not change multiplicatively beyond a
prespecified factor when the underlying user data changes.
With the proliferation of user data accumulated using such locally
private mechanisms, there is an increasing demand for designing data
analytics toolkits for operating on the collated user data.  In this
paper, we consider the design of algorithms aimed at providing a basic
ability to such a toolkit, namely the ability to run statistical tests
for the underlying user data distribution. At a high-level, we seek to
address the following question.
\begin{quote}\it
How should one conduct statistical testing on the (sensitive) data of
users, such that each user maintains their own privacy both to the
outside world \emph{and} to the (untrusted) curator performing the
inference?
\end{quote}

In particular, we consider two fundamental statistical inference
problems for a discrete distribution over a large alphabet:
\emph{identity testing} (goodness-of-fit) and \emph{independence
  testing}. A prototypical example of the former is testing whether
the user data was generated from a uniform distribution; the latter
tests if two components of user data vectors are independent. Our main
focus is the uniformity testing problem and most of the other results
are obtained as an extension using similar techniques.  We seek
algorithms that are efficient in the number of LDP user data samples
required and can be implemented practically. These two problems are
instances of \emph{distribution testing}, a sub-area of statistical
hypothesis testing focusing on small-sample analysis introduced by
Batu et al.~\cite{BatuFRSW00} and Goldreich, Goldwasser, and
Ron~\cite{GoldreichGR98}.

Our results are comprehensive, and organized along two axes: First, we
consider tests that use existing LDP data release mechanisms to
collect inputs at the center and perform a post-processing test on
this aggregated data. Specifically, we consider the popular {\Rappor}
mechanism of \cite{ErlingssonPK14} and the recently introduced the
{Hadamard Response} mechanism (\HR) of~\cite{AcharyaSZ18}.  Because
these mechanisms have utility beyond our specific use-case of
distribution testing~--~\Rappor, for instance, is already deployed in
many applications~--~it is natural to build a more comprehensive data
analytics toolkit using the data accumulated by these mechanisms.  To
this end, we provide uniformity testing algorithms with optimal sample
complexity for both mechanisms; further, for \HR, we also provide an
independence testing algorithm and analyze its performance.

Second, we consider the more general class of public-coin mechanisms
for solving testing problems which are allowed to use public
randomness.  We present a new response mechanism, \emph{Randomized
  Aggregated Private Testing Optimal Response} (\Raptor), that only
requires users to send a \emph{single} privatized bit indicating
whether their data point is in a (publicly known) random subset of the
domain.  Using \Raptor, we obtain simple algorithms for uniformity and
independence testing that are sample-optimal even among public-coin
mechanisms.

We next provide a detailed description of our results, followed by a
discussion of the relevant literature to put them in perspective. At
the outset we mention that the problems studied here have been
introduced earlier in~\cite{Sheffet18, GaboardiR17}. Our 
algorithms outperform their counterparts from these papers, and we
complement them with information-theoretic lower bounds establishing
their optimality (except for the proposed \HR-based independence
test).

\subsection{Algorithms and results}
The privacy level of a locally private mechanism is often
parameterized by a single parameter $\eps>0$. Specifically, an
$\eps$-LDP mechanism ($cf.$ Duchi et al.~\cite{DuchiJW13}) ensures
that for any two distinct values of user data, the distribution of the
output reported to the curator is within a multiplicative factor of
$e^{\eps}$; smaller values of $\eps$ indicate stronger privacy
guarantees.  In this work, we focus on the high-privacy regime, and
assume throughout that $\eps\in(0,1]$; however, our choice of $1$ as
  an upper bound is to set a convention and can be replaced with any
  constant.

In uniformity testing, the user data comprises independent samples
from an unknown $\ab$-ary distribution. These samples are then made
available to the curator through an $\eps$-LDP mechanism, and she
seeks to determine if the underlying distribution was uniform or
$\dst$-far from uniform in total variation distance. How many locally
private samples must the curator access?

First, we consider two representative locally private mechanisms,
\Rappor and \HR. We briefly describe these mechanisms
here informally and provide a more complete definition
in~\cref{sec:preliminaries}. In \Rappor, the $\ab$-ary observation of
the user is first converted to a $\ab$-length vector using one-hot
encoding, and then each bit of this vector is independently flipped
with probability $ 1/(1+e^\eps)$. \HR, on the other
hand, is a generalization of the classic Randomized
Response (\RaR)~\cite{Warner65} which roughly maps each $\ab$-ary
observation $x$ to either a randomly chosen $+1$ entry of the $x$-th
row of the $\ab \times \ab$ Hadamard matrix with probability
$e^\eps/(1+e^\eps)$, or to a randomly chosen $-1$ entry with
probability $1/(1+e^\eps)$. Interestingly, both these mechanisms have
been shown recently to be sample-optimal for learning $\ab$-ary
distributions; see~\cite{DuchiJW17, ErlingssonPK14,WangHWNXYLQ16, YeB17,
  KairouzBR16, AcharyaSZ18}. Further, note that both \Rappor and \HR are
private-coin mechanisms, and are symmetric across users.

We propose the following algorithm to enable uniformity testing using
data obtained via \Rappor. Once again, the description here is brief
and a formal description is provided in~\cref{sec:rappor}.
\begin{algorithm}[H]
\begin{algorithmic}[1]
\State Obtain $Z_1,\dots, Z_\ns$ using \Rappor.
\State For each $x$
in $[\ab]$, compute the number $N_x$ of $\ab$-bit vectors $Z_i$ for
which the $x$-th entry is $1$.
\State Compute the test statistic $T$
described in \eqref{eq:rappor:z} which is, in essence, a
bias-corrected version of the collision statistic $\sum_x (N_x^2-
N_x)$.
\State If $T$ is more than roughly $\ns^2\dst^2\eps^2 /\ab$,
declare \textsf{uniform}; else declare \textsf{not uniform}.
\end{algorithmic}
\caption{Uniformity testing using \Rappor}
\end{algorithm}
\noindent We analyze the sample complexity of the above test and show
that it is order-wise optimal among all tests that use \Rappor.
\begin{result}[Sample complexity of uniformity testing using \Rappor]
The uniformity test described above requires
$O(\ab^{3/2}/(\dst^2\eps^2))$ samples. Furthermore, any test using
\Rappor must use $\Omega(\ab^{3/2}/(\dst^2\eps^2))$ samples.
\end{result}
Moving now to \HR, denote by $\q^*$ the output distribution of \HR
when the underlying samples are generated from the uniform
distribution. (Note that $\q^*$ can be computed explicitly.)  Invoking
Parseval's theorem, we show that the $\lp[2]$ distance between the
$\q^\ast$ and the output distribution of $\HR$ is roughly
$\eps/\sqrt{\ab}$ times the $\lp[2]$ distance between the uniform and
the user data distributions. This motivates the following test.
\begin{algorithm}[H]
\begin{algorithmic}[1]
\State Obtain $Z_1,\dots, Z_\ns$ using \HR.  \State Using an
appropriate $\lp[2]$-test, test if the $\lp[2]$ distance between the
distribution of $Z_i$'s and $\q^\ast$ is less than roughly
$\gamma\eps/\ab$; in this case declare \textsf{uniform}. Else declare
\textsf{not uniform}.
\end{algorithmic}
\caption{Uniformity testing using \HR}
\end{algorithm}
\noindent Our next result shows that this test is indeed
sample-optimal among all tests using \HR.
\begin{result}[Sample complexity of uniformity testing using \HR]
The uniformity test described above requires
$O(\ab^{3/2}/(\dst^2\eps^2))$ samples. Furthermore, any test using \HR
must use $\Omega(\ab^{3/2}/(\dst^2\eps^2))$ samples.
\end{result}

Both tests proposed above thus provably cannot be improved beyond this
barrier of $\Omega(\ab^{3/2}/(\dst^2\eps^2))$ samples. Interestingly,
this was conjectured by Sheffet to be the optimal sample complexity of
locally private uniformity testing~\cite{Sheffet18}, although no
algorithm achieving this sample complexity was provided. Yet, our next
result shows that one can achieve the same guarantees with much fewer
samples when public randomness is allowed.

Specifically, we describe a new mechanism \Raptor, described below:
\begin{algorithm}[H]
\begin{algorithmic}[1]
\State The curator and the users sample a uniformly random subset $S$
of $[\ab]$ of cardinality $\ab/2$.  \State Each user computes the bit
indicator $B_i=\indic{X_i \in S}$ and sends it using $\RaR$, $i.e.$, flips it with probability
$1/(1+e^\eps)$ and sends the outcome to the curator.
\end{algorithmic}
\label{alg:raptor}
\caption{The \Raptor mechanism}
\end{algorithm}
\noindent The key observation is that when the underlying distribution
is $\dst$-far from uniform, the bias of $B_i$ is $1/2 +
\Omega(\dst/\sqrt{\ab})$ with constant probability (over the choice of
$S$); while clearly, under uniform the bits $B_i$ are unbiased. Thus,
we can simply test for uniformity by learning the bias of the bit up
to an accuracy of $\dst/\sqrt{\ab}$, which can be done using
$O(\ab/(\dst^2\eps^2))$ samples from \Raptor. In fact, we further show
that (up to constant factors) this number of samples cannot be
improved upon.

\begin{result}[Sample complexity of locally private uniformity testing]
Uniformity testing using \Raptor requires $O(\ab/(\dst^2\eps^2))$
samples. Furthermore, any public-coin mechanism for locally private
uniformity testing requires $\Omega(\ab/(\dst^2\eps^2))$ samples.
\end{result}

Although we have stated the previous three results for uniformity testing,
our proofs extend easily to identity testing, $i.e.$, the problem of testing equality
of the underlying distribution to a fixed known distribution $\q$ which is not necessarily uniform.
In fact, if we allow simple preprocessing of user observations before applying locally private mechanisms, a reduction argument
due to Goldreich~\cite{Goldreich16} can be used to directly convert identity testing to uniformity testing.\medskip

Our final set of results are for independence testing, where user data
consists of two-dimensional vectors $(X_i, Y_i)$ from $[\ab]\times[\ab]$. We seek to
ascertain if these vectors were generated from an independent
distribution $\p_1\otimes\p_2$ or a distribution that is $\dst$-far in total
variation distance from every independent distribution. 
 For this problem, a natural counterpart of \Raptor which simply applies \Raptor
to each of the two coordinate using independently generated sets
yields a sample optimal test -- indeed, we then simply need to test if the pair of
indicator-bits are independent or not. This can be done using
$O(\ab^2/(\dst^2\eps^2))$, leading to the following result.

\begin{result}[Sample complexity of locally private independence testing]
The sample complexity of locally private independence testing is
$\Theta(\ab^2/(\dst^2\eps^2))$ and is achieved by a simple
public-coin mechanism that applies \Raptor to each coordinate of user
data. 
\end{result}

For completeness, we also present a private-coin mechanism for
independence testing based on \HR which requires
$O(\ab^{3}/(\dst^2\eps^4))$ samples. The proposed test builds on a
technique introduced in Acharya, Daskalakis, and Kamath~\cite{AcharyaDK15}
and relies on learning in $\chi^2$ divergence. Although this result is suboptimal in
the dependence on the privacy parameter $\eps$, it improves on
both~\cite{Sheffet18} and the testing-by-learning baseline approach by a factor of
$\ab$. We summarize all our results in~\cref{table:results} and compare them
 with the best known prior bounds from~\cite{Sheffet18}.
\begin{table}[ht]\centering
\newcommand{\symmetric}{{\small$\circlearrowright$}\xspace}
\newcommand{\asymmetric}{{\small$\not\circlearrowright$}\xspace}
\begin{tabular}{|c|c|c|c|}\hline
\multirow{2}{*}{} & \multicolumn{2}{|c|}{\textbf{This work}} & \textbf{Previous}~\cite{Sheffet18} \\\cline{2-4}
& Private-Coin & Public-Coin & Private-Coin \\\hline
Uniformity Testing & $O\Paren{\frac{\ab^{3/2}}{\dst^2\eps^2} }$ \symmetric{} {\small$\star$} & $\Theta\Paren{\frac{\ab}{\dst^2\eps^2} }$ \symmetric & $O\Paren{\frac{\ab^{2}}{\dst^2\eps^2} }$ \asymmetric{} \\\hline
Independence Testing & $O\Paren{\frac{\ab^{3}}{\dst^2\eps^4} }$ \symmetric & $\Theta\Paren{\frac{\ab^{2}}{\dst^2\eps^2} }$ \symmetric & $O\Paren{\frac{\ab^{4}}{\dst^2\eps^2} }$ \asymmetric{}\\\hline
\end{tabular}
\caption{Summary of our results and previous work. The independence testing results hold for independence testing of distributions over $[\ab]\times[\ab]$; the symbol \symmetric (resp. \asymmetric) indicates a symmetric (resp. asymmetric) mechanism. Finally, $\star$ indicates that the upper bound is tight (in all parameters) for the subclass of private mechanisms our mechanisms belong to.}\label{table:results}
\end{table}

\subsection{Proof techniques}
We start by describing the analysis of our tests based on existing
$\eps$-LDP mechanisms. 
Recall that a standard (non-private) uniformity test entails estimating the $\lp[2]$
norm of the underlying distribution by counting the number of
collisions in the observed samples. When applying the same
idea on the data collected via \Rappor, we can naively try to estimate
the number of collisions by adding the number of pairs of output vectors with
$1$s in the $x$-th coordinate, for each $x$. However, the resulting statistic 
has a prohibitively high variance stemming from the noise added by \Rappor. We fix this
shortcoming by considering a bias-corrected version of this statistic
that closely resembles the classic $\chi^2$ statistic. However,
analyzing the variance of this new statistic turns out to be rather
technical and involves handling the covariance of quadratic functions
of correlated binomial random variables. Our main technical effort in
this part goes into analyzing this covariance, which may find further applications.

For our second test that builds on \HR, we follow a different
approach. In this case, we exploit the structure of Hadamard transform
and take recourse to Parseval's theorem to show that the $\lp[2]$
distance to uniformity of the original distribution $\p$ is equal, up
to an $\eps/\sqrt{\ab}$ factor, to the $\lp[2]$ distance of the
Fourier transform  $H(\p)$ to some (explicit) fixed distribution $\q$;
further, it can be shown that $\normtwo{\q} = O(1/\sqrt{\ab})$. With
this structural result in hand, we can test identity of $H(\p)$ to
$\q$ in the Fourier domain, by invoking the non-private $\lp[2]$
tester of Chan et al.~\cite{ChanDVV14} with the corresponding distance
parameter $\dst\eps/\sqrt{\ab}$. Exploiting the fact that $\q$ has a
small $\lp[2]$ norm leads to the stated sample complexity. 

Our private-coin mechanism for independence testing uses \HR as well, and
once again hinges on the idea that testing and learning in the 
Fourier domain can be done efficiently. To wit, we adapt the
 ``testing-by-learning'' framework of Acharya, Daskalakis,
and Kamath~\cite{AcharyaDK15} (which they show can be applied to many
testing problems, including independence testing) to our
 private setting. The main insight here is that
 instead of using \HR to learn and test the original distribution $\p$ in
$\chi^2$ distance, we perform both operations directly in the transformed
domain to the distribution at the output of \HR. Namely, we first
learn the transform of $\p_1\otimes\p_2$, then test 
whether the outcome is close to the transform of $\p$. The main
challenge here is to show that the variant of 
Hadamard transform that we use preserves (as was the case for uniformity
testing) the $\lp[2]$ distance from independence. We believe this
approach to be quite general, as was the case in~\cite{AcharyaDK15}, and that it can be used to tackle
many other distribution testing questions such as locally private
testing of monotonicity or log-concavity.

As mentioned above, our main results -- the optimal public-coin mechanisms for identity
and independence testing -- are  remarkably simple. The key heuristic underlying both can be summarized as follows: \emph{If $\p$ is $\dst$-far from uniform, then with constant probability a uniformly random subset $S\subseteq[\ab]$ of size $\ab/2$ will satisfy 
$\p(S) = 1/2 \pm \Omega(\dst/\sqrt{\ab})$}; on the other hand, if $\p$
is uniform then $\p(S)=1/2$ always holds. Thus, one can reduce the
original testing problem (over alphabet size $\ab$) to the much
simpler question of estimating the bias of a coin. This latter task is
very easy
to perform optimally in a locally private manner~--~for instance it
can be completed via \RaR~--~and requires  each player to send only
\emph{one} bit to the server. Hence, the main technical difficulty is
to prove this quite intuitive claim. We do this by showing
anticoncentration bounds for a suitable random variable by bounding
its fourth moment and invoking the Paley--Zygmund inequality. As a
byproduct, we end up establishing a more general version,
\cref{theorem:random:subset}, which we believe to be of independent
interest. 
\medskip

Our information-theoretic lower bounds are all based on a general
approach introduced recently by Acharya, Canonne, and
Tyagi~\cite{ACT:18} (in a non-private setting) that allows us to
handle the change in distances 
between distributions when information constraints are imposed on
samples. We utilize the by-now-standard ``Paninski
construction''~\cite{Paninski08}, a collection $\mathcal{C}$ of
$2^{\ab/2}$ distributions obtained by adding a small pointwise
perturbation to the $\ab$-ary uniform distribution. In order to obtain
a lower bound for the sample complexity of locally private uniformity
testing, following~\cite{ACT:18}, we identify such a
mechanism to the $\ns$ noisy channels $(\mech_j\colon
[\ab]\to\{0,1\}^\ast)_{j\in[\ns]}$ (that is, the randomized mappings
used by the $\ns$ players) it induces on the samples and consider the
distribution  $\mathcal{\mech}(\p)$ of the tuple of $\ns$ messages when the
underlying distribution of the samples is $\p$. The key step then is
to bound the $\chi^2$ divergence between 
(i)~$\mathcal{\mech}(\uniform)$, the distribution of the messages
under the uniform distribution; and
(ii)~$\shortexpect_{\p\in\mathcal{C}}[\mathcal{\mech}(\p)]$, the
\emph{average} distribution of the messages when $\p$ is chosen
uniformly at random among the ``perturbed distributions.''

Using the results of~\cite{ACT:18}, this in turn is tantamount to
obtaining an upper bound the Frobenius norm of specific
$[\ab/2]\times[\ab/2]$ matrices $\matH_1,\dots,\matH_\ns$ that capture
the information constraints imposed by $\mech_j$'s. Deriving these
bounds for Frobenius norms constitutes the main technical part of the lower bounds and relies on a careful analysis of the underlying mechanism and of the LDP constraints it must satisfy.

\paragraph{On the range of parameters.} As pointed out earlier, in
this work we focus on the high-privacy regime, $i.e.$, the case when
the privacy parameter $\eps$ is small and the privacy constraints on
the mechanisms are the most stringent. From a technical standpoint,
this allows us to rewrite the expressions such as
$\frac{e^\eps-1}{e^\eps+1}$ and $\frac{e^{\eps/2}-1}{e^{\eps/2}+1}$,
which  appear frequently, as simply $\Theta(\eps)$ and greatly
simplifies the statements of our results. However, our results carry
through to the general setting of large $\eps$, 
with $\frac{e^\eps-1}{e^\eps+1}$ replacing
$\Theta(\eps)$ term; the former is $\Theta(1)$ for large $\eps$.

\subsection{Related prior work}
Testing properties of a distribution by observing samples from it is a central problem in statistics and has been studied for over a century. Motivated by applications arising from algorithms dealing with massive amounts of data, it has seen renewed interest in the computer science community under the broad title of \emph{distribution testing}, with a particular focus on sample-optimal algorithms for discrete distributions. This literature itself is over two decades old; we refer an interested reader to surveys and 
books~\cite{Rubinfeld12,Canonne15,Goldreich17,BalakrishnanW17} for a comprehensive review. Here, we only touch upon works that are related directly to our paper.

Sample complexity for uniformity testing was settled
in \cite{Paninski08}, following a long line of work. The related, and
more general, problem of identity testing has seen revived interest
lately. The sample complexity for this problem was shown to be
$\Theta(\ab^{1/2}/\dst^2)$ in~\cite{ValiantV17a}, and by now even the
optimal dependence on the error probability is known ($cf.$~\cite{HuangM13,
DiakonikolasGPP16}). Moreover, a work of Goldreich~\cite{Goldreich16}
further shows that any uniformity testing algorithm implies
an identity testing one with similar sample complexity. Another
variant of this problem, termed ``instance-optimal'' identity testing
and introduced in~\cite{ValiantV17a},
seeks to characterize the dependence of the sample complexity on the
distribution $\q$ we are testing identity to, instead of the alphabet size. As pointed out
in~\cite{ACT:18}, the reduction from~\cite{Goldreich16} can be used in
conjunction with results from~\cite{BlaisCG17} to go through even
for the instance-optimal setting. This
observation allows us to focus on uniformity testing only, even when
local privacy constraints are imposed.

The optimal sample complexity for the independence testing problem
 where both observations are from the same set\footnote{The more general question asks to
 test independence of distributions over $[\ab_1]\times[\ab_2]$, or
 even over $[\ab_1]\times\dots\times[\ab_d]$. Optimal (non-private)
 sample complexities for these generalizations are also
 known~\cite{DiakonikolasK16}.}{}  
$[\ab]$ was shown to be $\Theta(\ab/\dst^2)$ in~\cite{AcharyaDK15,DiakonikolasK16}.

Moving now to distribution testing settings with privacy constraints, 
the setting of {\em differentially private} (DP) testing has by now been extensively studied. Here the algorithm itself is run by a trusted curator who
has access to all the user data, but needs to ensure that the output
of the test maintains differential privacy. Private
identity testing in this sense has been considered in~\cite{CaiDK17,
AliakbarpourDR17}, with a complete characterization of sample complexity
derived in~\cite{AcharyaSZ17}. Interestingly, in several parameter
ranges of interest the sample complexity here matches the sample
complexity for the non-private case
 discussed earlier,
showing that ``privacy often comes at no additional cost'' in this setting. As we show
in this work, this is in stark contrast to what can be achieved in the
more stringent locally private setting.

We are not aware of any existing private algorithm for DP independence
testing. While the literature on DP testing
includes several interesting mechanisms, for instance the
works~\cite{GaboardiLRV16, KiferR17, WangLK15} which contain
mechanisms for both identity and
independence testing, finite-sample
guarantees are not available and the results hold only in the
asymptotic regime. 

Finally, coming to the literature most closely related to our work,
locally private hypothesis testing was
considered first by Sheffet in~\cite{Sheffet18} where, too, both identity and independence
testing were considered. This work characterized the sample complexity
of LDP independence and uniformity testing when using Randomized Response, and introduced more general mechanisms. However, as
pointed-out in~\cref{table:results}, the 
algorithms proposed in~\cite{Sheffet18} require significantly more samples than
our sample-optimal algorithms for those questions. Moreover, the overall
sample complexity without restricting to any specific class of
mechanisms has not been considered.

An interesting concern studied in 
Sheffet's work is the distinction between \emph{symmetric}
and \emph{asymmetric} mechanisms. Broadly speaking, the latter are
locally private mechanisms where each player applies the same
randomized function $\mech$ to its data, where asymmetric mechanisms allow
different behaviors, with player $i$ using its own $\mech_i$. While we
mention this distinction in our results (see~\cref{table:results}), we
observe in~\cref{fact:asymmetric:advantage} that allowing asymmetric
mechanisms can only improve the sample complexity by at most a logarithmic
factor.

Another class of problems of statistical inference requires learning
the unknown distribution up to a desired accuracy of $\dst$ in total
variation distance. Clearly, the testing problems we consider can be solved
by privately learning the distributions (to accuracy
$\dst$). The optimal sample complexity of 
locally private learning discrete $\ab$-ary distributions is known to
be $\Theta(\ab^2/(\dst^2\eps^2))$; see~\cite{DuchiJW17, ErlingssonPK14, YeB17, KairouzBR16,
AcharyaSZ18}. (Furthermore, all these sample-optimal learning schemes are
symmetric.) This readily implies a sample complexity upper bound of
$O(\ab^2/(\dst^2\eps^2))$ for locally private identity testing, and of
$O(\ab^4/(\dst^2\eps^2))$ for independence testing. In this respect
the theoretical guarantees from~\cite{Sheffet18} are either implied or
superseded by this ``testing-by-learning'' approach.

\section{Notation and Preliminaries}\label{sec:preliminaries}
We write $[\ab]$ for the set of integers $\{1,2,\dots,\ab\}$, and
denote by $\log$ and $\ln$ the binary and natural logarithms,
respectively. We make extensive use of the standard asymptotic
$O(\cdot)$, $\Omega(\cdot)$, and $\Theta(\cdot)$ notation; moreover,
we shall sometimes use $a_n \lesssim b_n$, $a\gtrsim b_n$, and $a_n
\asymp b_n$ for their non-asymptotic counterparts (\textit{i.e.}, $a_n \lesssim
c_1 b_n$, $a\gtrsim c_1 b_n$, and $c_1 a_n \leq b_n \leq c_2 a_n$ for
every $n$, where $c_1,c_2>0$ are absolute constants).

Following the standard setting of distribution testing, we consider
probability distributions over a discrete (and known) domain $\Omega$.
Denote by $\distribs{\Omega}$ the set of all such distributions,
\[
    \distribs{\Omega} = \setOfSuchThat{ \p\colon [\ab] \to [0,1] }{
      \sum_{x\in\Omega} \p(x) = 1 },
\]
endowed with the \emph{total variation distance} (statistical
distance) as a metric, defined as $\totalvardist{\p}{\q} =
\sup_{S\subseteq \Omega}(\p(S)-\q(S))$.  It is easy to see that
$\totalvardist{\p}{\q} = \frac{1}{2}\normone{\p-\q}$, where
$\normone{\p-\q}$ is the $\lp[1]$ distance between $\p$ and $\q$ as
probability mass functions.  For a distance parameter $\dst\in(0,1]$,
  we say that $\p,\q\in\distribs{\Omega}$ are \emph{$\dst$-far} if
  $\totalvardist{\p}{\q}>\dst$; otherwise, they are
  \emph{$\dst$-close}.  We denote by $\p_1\otimes\p_2$ the product
  distribution over $[\ab_1]\times[\ab_2]$ defined by
  $(\p_1\otimes\p_2)(x_1,x_2)=\p_1(x_1) \cdot \p_2(x_2)$, for
  $\p_1\in\distribs{[\ab_1]}$, $\p_2\in \distribs{[\ab_2]}$.

In distribution testing, for a prespecified set of distributions
$\mathcal{C}\subseteq\distribs{\Omega}$ and given independent samples
from an unknown $\p\in\distribs{\Omega}$, our goal is to distinguish
between the cases (i) $\p\in \mathcal{C}$ and (ii) $\p$ is $\dst$-far
from every $\q\in\mathcal{C}$ with constant probability\footnote{As is
  typical, we set that probability to be $2/3$; by a standard
  argument, this can be amplified to any $1-\delta$ at the price of an
  extra $O(\log(1/\delta))$ factor in the sample complexity and
  running time.}.  The \emph{sample complexity of testing
  $\mathcal{C}$} is defined as the minimum number of samples required
to achieve this task in the worst case over all
$\p\in\distribs{\Omega}$ (as a function of $\dst$, $\abs{\Omega}$, and
all other relevant parameters of $\mathcal{C}$).\medskip

The specific problem of \emph{identity testing} corresponds to
$\Omega=[\ab]$ and $\mathcal{C} = \{\q\}$ for some fixed and known
$\q\in\distribs{[\ab]}$. {\em Uniformity testing} is the special case
of identity testing with $\q$ being the uniform distribution,
\textit{i.e.}, $\q(x) = 1/k$ for all $x \in [k]$. Lastly, \emph{independence
  testing} corresponds to $\Omega=[\ab]\times[\ab]$ and $ \mathcal{C}
\eqdef \setOfSuchThat{ \p_1\otimes\p_2 }{ \p_1,\p_2\in\distribs{[\ab]}
} $.

\subsection{Local Differential Privacy}
We consider the standard setting of $\eps$-local differential privacy, 
which we recall below.  A $1$-user mechanism is simply a randomized
mapping which, given as input user data $x\in \cX$, outputs a random
variable $Z$ taking values in $\cZ$. We represent this mechanism by a
channel $\mech\colon \cX\to \cZ$ where $\mech(z\mid x)$ denotes the probability that
the mechanism outputs $z$ when the user input is $x$. Similarly, an
$\ns$-user mechanism is represented by
$\mech=(\mech_j\colon\cX\to\cY)_{j\geq 0}$ where $\mech_j$ denotes the
channel used for the $j$-th user; when $\ns$ is clear from context, we
will simply use \emph{mechanism} for an $\ns$-user mechanism.  
For our purposes, $\cX$ will be the domain of our discrete probability
distributions, $[\ab]$, and $\cZ$ will be identified with
$\{0,1\}^\numbits$, for some integer $\numbits\geq 0$.

Note that each channel $\mech_j$ is applied independently to each
user's data. In particular, for independent samples $X_1,\dots, X_\ns$,
the outputs $Z_1, \dots, Z_\ns$ of $\mech$ are independent, too. The
mechanisms described above are \emph{private-coin} mechanisms: they only
require independent, local randomness at each user to implement the
local channels $\mech_1,\dots,\mech_\ns$. A private-coin mechanism is further said to be \emph{symmetric} if $\mech_j$
is the same for all $j$, in which case, with an abuse of notation, we
denote it $\mech\colon\cX\to\cZ$. A broader class of mechanisms of interest to us
are \emph{public-coin} mechanisms, where the output of each user may
depend additionally on shared public randomness $U$ (independent of the users' data); when the shared
randomness takes the value $u$, the mechanism uses channels
$\mech^u_j$. Clearly, private-coin mechanisms are a special case, corresponding to constant
$U$. The above distinction between symmetric and asymmetric mechanisms applies to public-coin mechanisms as well. 

A public-coin mechanism $\mech$ is an \emph{$\eps$-locally differentially private}
($\eps$-LDP) mechanism if it satisfies the following:
\begin{equation}
    \max_{u}\max_{z\in \cY}\max_{x,x'\in \cX}\frac{\mech_j^u(z\mid
      x')}{\mech_j^u(z\mid x)} \leq e^{\eps}, \quad   \forall \,1\leq
    j\leq \ns.   
\end{equation}

\subsection{Existing LDP mechanisms}

Three LDP mechanisms will be of interest to us: \emph{randomized
  response}, \emph{\Rappor}, and \emph{Hadamard response}.\vspace{-0.75em}

\paragraph{Randomized response.}
The \emph{$\ab$-randomized response}  ($\ab$-\RaR)
mechanism~\cite{Warner65} is an $\eps$-LDP mechanism, $\mech_{\RaR}$,
with $\cZ=\cX=[\ab]$, such that
\begin{equation}
\label{eq:rr}
\mech_\RaR(z \mid x) \eqdef \begin{cases} \frac{e^\eps}{e^\eps+\ab-1}
  & \text{if $z=x$},\\ \frac{1}{e^\eps+\ab-1} & \text{otherwise.}
\end{cases}
\end{equation}
Originally introduced for the binary case ($\ab=2$), it is 
one of the simplest and most natural response mechanisms.\vspace{-0.75em}

\paragraph{\Rappor.}
The \emph{randomized aggregatable privacy-preserving ordinal response}
(\Rappor) is an $\eps$-LDP mechanism introduced
in~\cite{DuchiJW13,ErlingssonPK14}.  Its simplest implementation,
$\ab$-\Rappor, maps $\cX=[\ab]$ to $\cZ=[2^\ab]$ in two steps.  First, a
one-hot encoding is applied to the input $x\in[\ab]$ to obtain a
vector $y\in\{0,1\}^\ab$ such that $y_{j}=1$ for $j=x$ and $\by_j=0$
for $j\ne x$.  The privatized output, $z \in \cZ$, of $\ab$-\Rappor is
represented by a $\ab$-bit vector obtained by independently flipping
each bit of $y$ independently with probability
$\frac{1}{e^{\eps/2}+1}$.

Note that if $x$ is drawn from $\p\in\distribs{\ab}$, this leads to
$z\in\{0,1\}^\ab$ such that the coordinates are (non-independent)
Bernoulli random variables with $z_i$ distributed as
$\bernoulli{\alpha_{R}\cdot\p(i)+\beta_{R}}$ where $\alpha_{R}, \beta_{R}$
are defined as
\begin{equation}\label{eq:rappor:parameters}
    \alpha_{R} \eqdef \frac{e^{\eps/2}-1}{e^{\eps/2}+1} =
    \frac{\eps}{4}+o(\eps), \qquad \beta_{R} \eqdef
    \frac{1}{e^{\eps/2}+1} = \frac{1}{2} + o(1).
\end{equation}

\paragraph{Hadamard Response.}
\label{sec:general-private}
Hadamard response is a symmetric, communication- and time-efficient
mechanism, proposed in~\cite{AcharyaSZ18}.

In order to define the Hadamard response mechanism, we first define a
general family of $\eps$-LDP mechanisms that include \RaR as a special
case.  Let $s\leq K$ be two integers, and for each $x\in \cX=[k]$ let $C_x
\subseteq \cZ=[K]$ be a subset of size with $|C_x|=s$.  Then, the
general privatization scheme is described by
\begin{equation}\label{eq:general:asz}
\forall z\in[K],\qquad \mech(z\mid x) \eqdef \begin{cases}
  \frac{e^\eps}{se^\eps+K-s} & \text{if $z\in
    C_x$},\\ \frac{1}{se^\eps+K-s} & \text{if $z\in\cZ\setminus C_x$}
                        \end{cases}
\end{equation}
which can easily be seen to be $\eps$-LDP.  Further, note that $k$-\RaR
corresponds to the special case with $K=k$, $s=1$, and $C_x = \{x\}$ for
all $x$. \medskip

The \emph{Hadamard Response mechanism} (\HR), is obtained by choosing
$s=K/2$, and a collection of sets $(C_x)_{x\in[\ab]}$ such that
\begin{enumerate}[(A)]
\item\label{item:condition:size} For every $x\in[\ab]$, $|C_{x}|= s =
  \frac{K}{2}$.
\item\label{item:condition:diff} For every distinct $x, x'\in[\ab]$,
  the symmetric difference $\Delta(C_{x}, C_{x'})$ satisfies
  $\abs{\Delta(C_{x}, C_{x'})} = s$.
\end{enumerate}
For these parameters, we get that
\begin{equation}\label{eq:hr:asz}
\mech(z\mid x) = \begin{cases}
  \frac{2}{K}\cdot \frac{e^\eps}{e^\eps+1} & \text{if $z\in
    C_x$},\\ \frac{2}{K}\cdot\frac{1}{e^\eps+1} & \text{if
    $z\in\cZ\setminus C_x$},
                  \end{cases}
\quad \forall z\in[K].
\end{equation}
Let $q(\p, C_x)$ denote the probability that the privatized output $z$
lies in $C_x$, when the input distribution is $\p$.
Then, from \ref{item:condition:size} and \ref{item:condition:diff} it
can be seen that
\begin{align*}
	\probaCond{ Z \in C_x }{ Z=x } &= \sum_{z\in C_x}\mech(z\mid
        x) 
        \frac{e^\eps}{e^\eps+1}, \qquad \probaCond{ Z \in C_x }{ Z\neq
          x } = \frac{1}{2}
\end{align*}
and combining these two
\begin{equation}\label{eqn:pcx2px}
	q(\p, C_x) = \frac{1}{2}+ \frac{e^\eps-1}{2(e^\eps+1)}
        \p(x)\,.
\end{equation}

A method for constructing sets $(C_x)_{x\in [\ab]}$ that also allows
efficient implementation of the resulting mechanism was proposed in \cite{AcharyaSZ18} 
using Hadamard codes (hence the name Hadamard Response). Specifically,
let 
\begin{equation}\label{eqn:choice:K}
	K \eqdef 2^{\ceil{\log(\ab+1)}}
\end{equation}
so that $\ab+1\leq K\leq 2(\ab+1)$, and let $H_K\in \{-1,1\}^{K\times
  K}$ be the Hadamard matrix of order $K$
(see~\cref{ssec:hadamard:prelims} for more details).  Hereafter, we identify each row of $H_K$ to a subset
of $[K]$.  As $K\geq \ab+1$, we can pick an injection
$\phi\colon[\ab]\to\{2,\dots,K\}$ and map each $x\in[\ab]$ to a
distinct subset $C_x\subseteq[K]$ defined by the $\phi(x)$-th row of
$H_K$.  By~\cref{fact:hadamard:matrix} in the next section, this family
$(C_x)_{x\in[\ab]}$ satisfies \ref{item:condition:size} and
\ref{item:condition:diff}.

\subsection{Hadamard matrices and linear codes}\label{ssec:hadamard:prelims}
Next, we recall some useful properties of Hadamard
matrices which will be needed for our analysis of \HR-based tests.
\begin{definition}
	\label{def:hadamard}
  Let $M\geq 1$ be any power of two. The \emph{Hadamard matrix of
    order $M$}, denoted $H_{M}$, is the matrix of size $M\times M$
  defined recursively by Sylvester's construction: (i)
  $H_1\eqdef\begin{bmatrix}1\end{bmatrix}$, and (ii) for $m \geq 1$,
	\[
	H_{2^m} \eqdef \begin{bmatrix} H_{2^{m-1}} & H_{2^{m-1}}
          \\ H_{2^{m-1}} & -H_{2^{m-1}} \end{bmatrix}\,.
	\]
\end{definition}
Note that all entries of $H_M$ are in $\{-1,1\}$.

\begin{fact}\label{fact:hadamard:matrix}
  Let $m\geq 1$ be any integer. Then, the Hadamard matrix $H_{2^m}$
  has the following properties:
  \begin{enumerate}[(i)]
    \item The first row of $H_{2^m}$ is the all-one vector.
    \item For every $j\geq 2$, the $j$-th row of $H_{2^m}$ is
      \emph{balanced}, i.e, contains exactly $2^{m-1}$ entries equal
      to $1$.
    \item Every two distinct rows are orthogonal; that is, for every
      $1\leq i<j\leq 2^m$, the $i$-th and $j$-th row agree
      (resp. disagree) on exactly $2^{m-1}$ entries.
  \end{enumerate}
\end{fact}

Fix any $m\geq 1$. The Hadamard matrix $H_{2^m}$ corresponds to the
\emph{Walsh--Hadamard transform} (or Fourier transform; see, for
example,~\cite{ODonnell14}). Specifically, for any two functions 
$f,g\colon \{0,1\}^m\to \RR$, define the inner product
$\dotprod{f}{g}$ over $\RR^{\{0,1\}^m}$ as
\begin{equation}
    \dotprod{f}{g} \eqdef \frac{1}{2^m}\sum_{x\in\{0,1\}^m} f(x)g(x),
\end{equation}
and let $\norm{\cdot}$ denote the norm induced by this inner product.
Moreover, the functions $\chi_S\colon \{0,1\}^m\to\RR$ defined for every
$S\subseteq [m]$ by 
 $\chi_S(x) =
(-1)^{\dotprod{x}{S}} = \prod_{i\in S} (-1)^{x_i}$ form an orthonormal basis, whereby
every $f\colon\{0,1\}^m \to \RR$ can be uniquely written as
\begin{equation}
   \forall x\in\{0,1\}^m, \qquad f(x) = \sum_{x\in\{0,1\}^m}
   \hat{f}(S) \chi_S(x)
\end{equation}
where $\hat{f}(S) \eqdef \dotprod{f}{\chi_S}$. The Walsh--Hadamard
matrix specifies this transformation of basis. Specifically, we 
note following standard fact:
\begin{fact}\label{fact:hadamard:fourier}
    Let $m\geq 1$. Then, for every $x\in\{0,1\}^m$ and
    subset $S\subseteq[m]$ identified to its characteristic vector
    $s\in\{0,1\}^m$, we have that
    \[
      (H_{2^m})_{s,x} = (H_{2^m})_{x,s} = \chi_S(x)\,.
    \]
\end{fact}
This spectral view of Walsh--Hadamard matrix leads to 
Parseval's Theorem, which is instrumental in design of our tests based
on \HR.
\begin{theorem}[Parseval's Theorem]\label{theo:parseval}
    For every function $f\colon\{0,1\}^m\to\RR$,
    \[
        \norm{f}^2 = \sum_{S\subseteq [m]} \hat{f}(S)^2\,.
    \]
\end{theorem}
\noindent 
\subsection{On symmetry and asymmetry}
While all the LDP mechanisms underlying our proposed sample-optimal tests in this paper can be cast as symmetric mechanisms, the next result shows that
asymmetric mechanisms can in any case yield at most a
logarithmic-factor improvement in sample complexity over symmetric
ones.
\begin{lemma}\label{fact:asymmetric:advantage}
Suppose that there exists a private-coin (respectively public-coin) LDP mechanism
for some task $\mathcal{T}$ with $\ns$ users and probability of
success $5/6$. Then, there exists a private-coin (respectively public-coin) 
\emph{symmetric} LDP mechanism for $\mathcal{T}$ with
$\ns'=O(\ns\log\ns)$ users and probability of success $2/3$.
\end{lemma}
\begin{proof}
Let $\mech=(\mech_i)_{i\in[\ns]}$ be the purported mechanism, with
$\mech_i\colon \cX\to\cY$ being the mapping of the $i$-th user. We
create a symmetric (randomized) mechanism
$\tilde{\mech}\colon\cX\to[\ns]\times\cY$ as follows: On input
$x\in\cX$, use private (respectively public) randomness to generate $I\in[\ns]$
uniformly at random (and independently of everything else); and output
$(I, \mech_I(x))$.\footnote{Note that for public-coin mechanisms, one
  can define $\tilde{\mech}\colon\cX\to\cY$, as there is no need for a
  user to communicate the random index $I$ to the referee.}

Clearly, the resulting mechanism is symmetric. Further, by a standard coupon-collector
argument, for $\ns'=O(\ns\log\ns)$ we have that with probability at
least $5/6$, each $i\in[\ns]$ will be drawn at least
once. Whenever this is the case, upon gathering all the outputs, the
referee can then select a subset of $\ns$ outputs and simulate the
original mechanism, having received the output of $\mech_1,\dots,
\mech_\ns$. Overall, the probability of failure is at most
$1/6+1/6=1/3$ by a union bound.
\end{proof}

\subsection{A warmup for the binary case}

We conclude this section with simple algorithms for identity and
independence testing for the case when $\Omega = \{0,1\}$, i.e, for
support size $\ab = 2$. These algorithms will be used later in our
optimal tests based on \Raptor.

\subsubsection{Private estimation of the bias of a coin.}
First, we deal with the problem of estimating the bias of a coin up to 
an additive accuracy of $\pm\dst$, when the outcomes of coin tosses
can be accessed via an $\eps$-LDP mechanism. Note that this yields as
a corollary an algorithm for identity testing over
$\{0,1\}$.  Indeed, to test if the generating distribution $\p$ equals $\q \in
\Delta(\{0,1\})$ or is $\dst$-far from it, we estimate
probability $\p(0)$ to additive $\pm \dst/2$ and compare it with
$\q(0)$. The following result is a folklore and is included for
completeness.

\begin{lemma}[Locally Private Bias Estimation, Warmup]
\label{lem:binary-testing}
For $\eps \in(0,1]$, an estimate of the bias of a coin with an
  additive accuracy of $\dst$ can be obtained using $
  O\!\Paren{1/(\dst^2\eps^2)} $ samples via $\eps$-LDP \RaR. Moreover, 
  any estimate of bias obtained via $\eps$-LDP \RaR must use
  $\Omega\!\Paren{1/(\dst^2\eps^2)}$ samples.
\end{lemma}
\begin{proof}
Recall from \eqref{eq:rr} that an $\eps$-LDP \RaR is described by
the channel $ \mech_{\RaR}(0\mid 0)=\mech(1\mid 1) =
\frac{e^\eps}{e^\eps+1} $.  When a $\bernoulli{\rho}$ random variable
passes through this channel, the output is a Bernoulli random variable
with mean
\[
\rho' \eqdef \rho\frac{e^\eps}{e^\eps+1}+(1-\rho)\frac{1}{e^\eps+1} =
\frac{1}{e^\eps+1} + \rho \frac{e^\eps-1}{e^\eps+1}.
\]
Therefore, estimating $\rho$ to $\pm\dst$ using this mechanism is
equivalent to estimating $\rho'$ to an additive $\dst'\eqdef
\frac{e^\eps+1}{e^\eps-1}\dst$, which can be done with
$O(1/{\gamma'}^2) = O(1/(\dst^2\eps^2))$ samples (the second as $\eps
\lesssim 1$).

It remains to prove optimality. For $\ab=2$, it can shown that any $\eps$-LDP scheme can be 
obtained by passing output of an $\eps$-LDP \RaR through another
channel. Therefore, \RaR will require the least number of samples for
estimating the bias, and it
suffices to show the claimed bound of
$\Omega\!\Paren{1/(\dst^2\eps^2)}$ for \RaR. To that end, suppose we provide as input a Bernoulli random variable with bias $1/2+\dst$ to
\RaR. Then, the output has bias 
$\frac{1}{2} +\dst\frac{e^\eps-1}{e^\eps+1} =
\frac{1}{2}+O(\gamma\eps)$. On the other hand, when the input is
$\bernoulli{1/2}$, then the 
output is $\bernoulli{1/2}$ as well. Therefore, distinguishing between
a $\bernoulli{1/2}$ and a $\bernoulli{1/2+\dst}$ using samples from an
$\eps$-LDP \RaR is at least as hard as distinguishing $\bernoulli{1/2}$ and
$\bernoulli{1/2 +O(\dst\eps)}$ without privacy constraints. This
latter task is known to require the stated number of samples.
\end{proof}
\subsubsection{Independence testing over $\{0,1\}\times\{0,1\}$}
As a corollary of~\cref{lem:binary-testing}, we obtain an algorithm for locally private independence testing for $\ab=2$, which, too, will be used later in the paper.
\begin{corollary}
\label{cor:binary-ind}
For $\eps\in (0,1]$, there exists a symmetric, private-coin $\eps$-LDP mechanism 
that tests
whether a distribution over $\{0,1\}\times\{0,1\}$ is a product
distribution or $\dst$-far from any product distribution 
using $O\!\Paren{1/(\dst^2\eps^2)} $ samples.
\end{corollary}
\begin{proof}
Consider a distribution $\p$ over $\{0,1\} \times \{0,1\}$ with
marginals $\p_1$ and $\p_2$.  Note that
\[
\abs{ \p(0,0)-\p_1(0)\p_2(0)|=|\p(x,y)-\p_1(x)\p_2(y) }, \quad x,y\in
\{0,1\}.
\]
Thus, if $\p$ is $\dst$-far in total variation distance from any
product distribution, it must hold that
$\totalvardist{\p}{\p_1\otimes\p_2}\geq \dst$, which in view of the
 equation  above yields $|\p(0,0)-\p_1(0)\p_2(0)|\geq \dst/2$. 
Using
this observation, we can test for independence using
$O(1/(\eps^2\dst^2))$ samples as follows. First, note that for any
symbol $x$, $\p(x)$ can be estimated up to an accuracy $\dst$ using
$O(1/(\eps^2\dst^2))$ samples by converting the observation $X$
to the binary observation $\indic{X=x}$ and applying the
estimator of \cref{lem:binary-testing}. Thus, we can estimate
$\p(0,0)$, $\p_1(0)$, and $\p_2(0)$ up to an accuracy $\dst/16$ by
assigning $O(1/(\eps^2\dst^2))$ samples each for them. Denote the
respective estimates by $\tilde{\p}(0,0)$, $\tilde{\p}_1(0)$, and
$\tilde{\p}_2(0)$. When $\p(0,0)=\p_1(0)\p_2(0)$,
\[
|\tilde{\p}(0,0)-\tilde{\p}_1(0)\tilde{\p}_2(0)|\leq
|\tilde{\p}(0,0)-\p(0,0)|+ \tilde{\p}_1(0)-\p_1(0)|+
|\tilde{\p}_2(0)-\p_2(0)| \leq \frac 3 {16} \dst.
\]
On the other hand, when $|\p(0,0)-\p_1(0)\p_2(0)|\geq \dst/2$, we have
\[
|\tilde{\p}(0,0)-\tilde{\p}_1(0)\tilde{\p}_2(0)| \geq
|\p(0,0)-\p_1(0)\p_2(0)|- |\tilde{\p}(0,0)-\p(0,0)|-
|\tilde{\p}_1(0)-\p_1(0)|- |\tilde{\p}_2(0)-\p_2(0)| \geq
\frac{5}{16}\dst.
\]
Thus, for $\ab=2$, locally private independence testing can be
performed with $O(1/(\eps^2\dst^2))$ samples by estimating the
probabilities $\p(0,0)$, $\p_1(0)$, $\p_2(0)$ and comparing $\abs{
  \tilde{\p}(0,0)-\tilde{\p}_1(0)\tilde{\p}_2(0) }$ to the threshold
$\dst/4$.
\end{proof}

\section{Locally Private Uniformity Testing using Existing Mechanisms}\label{sec:uniformity}
In this section, we provide two locally private mechanisms for
uniformity testing. 
As discussed earlier, this in turn provides similar mechanisms for identity testing as well. 
These two tests, based respectively on the symmetric, private-coin mechanisms \Rappor and \HR, will be seen to have the same sample complexity of $O(\ab^{3/2}/\dst^2\eps^2)$. However, the first has the advantage of being based on a
widespread mechanism, while the second is more
efficient in terms of both time and communication.   

\subsection{A mechanism  based on \Rappor}\label{sec:rappor}
Given $\ns$ independent samples from $\p$, let the output of \Rappor applied to these samples be denoted by $\bb_1, \ldots, \bb_\ns\in\{0,1\}^\ab$, where $\bb_i
= (\bb_{i1}, \ldots, \bb_{i\ab})$ for $i\in[\ns]$. The following fact is a simple consequence of the definition of \Rappor.
\begin{fact}\label{fact:rappor:statistics}
Let $i,j\in[\ns]$, and $x,y\in[\ab]$.
\[
    \probaOf{\bb_{ix} =1, \bb_{jy}=1} = \begin{cases}
    (\alpha_{R}\p(x)+\beta_{R})(\alpha_{R}\p(y)+\beta_{R}) & \text{ if
    } i\neq j\\
    (\alpha_{R}\p(x)+\beta_{R})(\alpha_{R}\p(y)+\beta_{R})-\alpha_{R}^2\p(x)\p(y)
    & \text{ if } i=j,\, x\neq y\\ \alpha_{R} \p(x) + \beta_{R}
    & \text{ if } i=j,\, x= y\\ \end{cases}
\]
where $\alpha_{R},\beta_{R}$ are defined as
in~\eqref{eq:rappor:parameters}.
\end{fact}

\paragraph{First idea: Counting Collisions.}

A natural idea would be to try and estimate $\normtwo{\p}^2$ by
counting the collisions from the output of \Rappor. Since this only adds post-processing to \Rappor, which is
LDP, the overall procedure does not violate the $\eps$-LDP constraint. For $\sigma_{i,j}^x$ defined as $\indic{\bb_{ix} =1, \bb_{jx}=1 }$, $x\in[\ab]$, $i\neq j$, the statistic
$S\eqdef \sum_{1\leq i<j\leq \ns} \sum_{x\in[\ab]} \sigma_{i,j}^x$
counting collisions over all samples and differentially private symbols can be seen to have expectation 
\[
    \expect{S}
    = \binom{\ns}{2}\Paren{\alpha_{R}^2\normtwo{\p}^2+2\alpha_{R}\beta_{R}
    + \ab\beta_{R}^2
    } \asymp \frac{1}{2}\eps^2\ns^2\normtwo{\p}^2+ \Theta(\ab)\,.
\]
Up to the constant normalizing factor, this suggests an unbiased
estimator for $\normtwo{\p}^2$, and thereby also for   
$\normtwo{\p-\uniform}^2 =\normtwo{\p}^2 - 1/\ab$.
However, the issue lies with the variance
of this estimator. Indeed, it can be shown that
$\variance{S} \approx \ns^3\ab$ (for constant $\eps$). Thus, if we use this statistic to distinguish between
$\normtwo{\p}^2=1/\ab$ and $\normtwo{\p}^2 >(1+\Omega(\dst^2))/\ab$ for uniformity testing, we need 
\[
    \sqrt{\ns^3\ab} \ll \ns^2 \eps^2\cdot \frac{\dst^2}{\ab}
\] 
$i.e.$, $\ns \gg \ab^3/(\dst^4\eps^4)$.  This sample requirement turns out to be off by a quadratic
factor, and even worse than the trivial upper bound obtained by learning $\p$.
\paragraph{An Optimal Mechanism.}
We now propose our testing mechanism based on \Rappor, which, in essence, uses
a privatized version of a $\chi^2$-type statistic
of~\cite{ChanDVV14,AcharyaDK15,ValiantV17a}. For $x\in[\ab]$, let the
number of occurrences of $x$ among the $\ns$ (privatized) outputs
of \Rappor be
\begin{equation}\label{eq:rappor:nx}
    N_x \eqdef \sum_{j=1}^\ns \indic{\bb_{jx}=1}
\end{equation}
which by the definition of \Rappor follows a
$\bin{\ns}{\alpha_{R}\p(x)+\beta_{R}}$ distribution. Now, letting
\begin{equation}\label{eq:rappor:z}
    T \eqdef \sum_{x\in[\ab]} \Paren{ \left(
    N_x-(\ns-1)\left(\frac{\alpha_{R}}{\ab}+\beta_{R}\right) \right)^2
    - N_x }
    + \ab(\ns-1)\left(\frac{\alpha_{R}}{\ab}+\beta_{R}\right)^2
\end{equation}
we get a statistic, applied to the output of \Rappor,
which (as we shall see) is up to normalization an unbiased estimator
for the squared $\lp[2]$ distance of $\p$ to uniform. The main
difference with the naive approach we discussed previously, however,
lies in the extra linear term. Indeed, the collision-based statistic
was of the form
\[
    S \propto \sum_{x\in[\ab]} \Paren{ N_x^2 - N_x },
\]
and in comparison, keeping in mind that $N_x$ is
typically concentrated around its expected value of roughly
$\ns/2$, our new statistics can be seen to take the form
\[
    T \approx \sum_{x\in[\ab]} \Paren{ N_x^2 - \ns N_x }
    + \Theta(\ab\ns^2),
\]
since $\beta_{R} \approx 1/2$. That is, now the fluctuations of the
quadratic term are reduced significantly by the subtracted linear
term, bringing down the variance of the statistic.

This motivates our testing algorithm based on \Rappor,~\cref{alg:rappor:chisquare}, and
leads to the main result of this section:
\begin{theorem}
\label{thm:ub-rappor}
For $\eps\in (0,1]$, \cref{alg:rappor:chisquare} based on  $\eps$-LDP \Rappor can test whether a distribution is uniform or $\dst$-far from uniform using 
\[O\Paren{\frac{\ab^{3/2}}{\dst^2\eps^2}}\] samples.
\end{theorem}

\begin{algorithm}[ht]
\begin{algorithmic}[1]
\Require Privacy parameter $\eps>0$, distance parameter $\dst\in(0,1)$, $\ns$ samples
\State Set
\[
    \alpha_{R} \gets \frac{e^{\eps/2}-1}{e^{\eps/2}+1}, \qquad \beta_{R} \gets \frac{1}{e^{\eps/2}+1}
\]
as in~\eqref{eq:rappor:parameters}.
\State Apply ($\eps$-LDP) \Rappor to the $\ns$ samples to obtain $(\bb_i)_{1\leq i\leq \ns}$ \Comment{Time $O(\ab)$ per user}
\State Compute $N_x$ for every $x\in[\ab]$, as defined in~\eqref{eq:rappor:nx} \Comment{Time $O(\ab \ns)$}
\State Compute $T$, as defined in~\eqref{eq:rappor:z} \Comment{Time $O(\ab)$}
\If{$T<\ns(\ns-1)\alpha_{R}^2\dst^2/\ab$}\label{alg:rappor:chisquare:threshold}
  \State \Return \textsf{uniform}
\Else
  \State \Return \textsf{not uniform}
\EndIf
\end{algorithmic}
\caption{Locally Private Uniformity Testing using \Rappor} \label{alg:rappor:chisquare}
\end{algorithm}

\begin{proofof}{\cref{thm:ub-rappor}}
Clearly, since \Rappor is an $\eps$-LDP mechanism, the
overall~\cref{alg:rappor:chisquare} does not violate the $\eps$-LDP
constraint. We now analyze the error performance of the proposed test, which we will do simply by using Chebyshev's
inequality. Towards that, we evaluate the expected value and the
variance of $T$.

The following evaluation of expected value of statistic $T$ uses a
simple calculation entailing moments of a Binomial random variable:
\begin{lemma}\label{lemma:ub-rappor:expect}
  With $T$ defined as above, we have \[ \expect{T}
    = \ns(\ns-1)\alpha_{R}^2\normtwo{\p-\uniform}^2 \] where the
    expectation is taken over the private-coins used by \Rappor and the samples drawn from $\p$. In particular, (i)~if
    $\p=\uniform$, then $\expect{T}=0$; while (ii)~if
    $\totalvardist{\p}{\uniform}>\dst$, then $\expect{T}>
    2\ns(\ns-1)\frac{\alpha_{R}^2\dst^2}{\ab}$.
\end{lemma}
\begin{proof}
    Letting
    $\lambda \eqdef \left(\frac{\alpha_{R}}{\ab}+\beta_{R}\right)$ and
    using the fact that $\expect{N_x^2} = \expect{N_x} +
    (1-\frac{1}{\ns})\expect{N_x}^2$, we
    have \begin{align*} \expect{T} &= \sum_{x\in[\ab]} \expect{ \left(
    N_x-(\ns-1)\lambda \right)^2 - N_x } + \ab(\ns-1) \lambda^2 \\
    &= \sum_{x\in[\ab]} \Paren{ \expect{N_x^2} -
    2(\ns-1)\lambda \expect{N_x} + (\ns-1)^2\lambda^2 - \expect{N_x} }
    + \ab(\ns-1) \lambda^2 \\
    &= \sum_{x\in[\ab]} \Paren{ \frac{\ns-1}{\ns}\expect{N_x}^2 -
    2(\ns-1)\lambda \expect{N_x} + \ns(\ns-1)\lambda^2 } \\
    &= \ns(\ns-1)\sum_{x\in[\ab]} \Paren{ \frac{\expect{N_x}^2}{\ns^2}
    - 2\lambda \frac{\expect{N_x}}{\ns} + \lambda^2 } \
    = \ns(\ns-1)\sum_{x\in[\ab]} \Paren{\frac{\expect{N_x}}{\ns}-\lambda}^2, \end{align*}
    which, along with the observation that $\frac{\expect{N_x}}{\ns}-\lambda
    = \alpha_{R}\Paren{\p(x)-1/\ab}$, gives the result.
\end{proof}

Turning to the variance, we get the following:
\begin{lemma}\label{lemma:ub-rappor:variance}
  With $T$ defined as above, we have \[ \variance{T} \leq 4\ab\ns^2 +
    8\ns^3 \alpha_{R}^2 \normtwo{\p-\uniform}^2 = 4\ab\ns^2 +
    8\ns\expect{T}\,.  \]
\end{lemma}
\noindent The proof of this lemma is quite technical and relies on a tedious analysis of the covariance of the random variables $(N_x)_{x\in[\ab]}$, in view of bounding quantities of the form $\cov(f(N_x), f(N_y))$. We defer the details to~\cref{app:variance:rappor}.\smallskip

With these two lemmata, we are in a position to conclude the
argument. Suppose $\ns \geq
C\cdot \frac{\ab^{3/2}}{\alpha_{R}^2\dst^2}$, for some 
constant $C>0$ to be specified later. Recall
that $\alpha_{R} = \eps/4+o(\eps)$ when $\eps \to 0$, leading to the
claimed sample complexity.

First, consider the case when $\p=\uniform$. In this case $\expect{T}=0$ and $\variance{T}\leq  4\ab\ns^2$
  by~\cref{lemma:ub-rappor:expect,lemma:ub-rappor:variance}, and so by
  Chebyshev's inequality \[ \probaOf{
  T \geq \ns^2\frac{\alpha_{R}^2\dst^2}{\ab}
  } \leq \frac{ \ab^2\variance{T} }{ \ns^4\alpha_{R}^4\dst^4
  } \leq \frac{17\ab^3}{8\ns^2\alpha_{R}^4\dst^4} < \frac{3}{C^2} \]
  which is at most $1/3$ for $C\geq 3$.  
 
Next, when $\totalvardist{\p}{\uniform}>\dst$, $\expect{T}>2\ns^2\frac{\alpha_{R}^2\dst^2}{\ab}$ and
  $\variance{T}\leq 4\ab\ns^2 + 8\ns \expect{T}$, and again by
  Chebyshev's inequality 
\[
 \probaOf{ T
  < \ns^2\frac{\alpha_{R}^2\dst^2}{\ab} } \leq \probaOf{ T
  < \frac{1}{2}\expect{T} } \leq \frac{ 4\variance{T} }{ \expect{T}^2
  } 
\leq \frac{17\ab^3}{2\ns^2\alpha_{R}^4\dst^4}
  + \frac{10\ab}{\ns\alpha_{R}^2\dst^2} \leq \frac{17}{2C^2}
  + \frac{10}{C\sqrt{\ab}} \] which is at most $1/3$ for $C\geq 23$. Taking $C = 23$ concludes the proof of~\cref{thm:ub-rappor}.
\end{proofof}

\subsection{A mechanism based on Hadamard Response}\label{sec:hadamard}
Although the \Rappor-based mechanism of~\cref{sec:rappor} achieves a
significantly improved sample complexity over the naive learning-and-testing approach, 
it suffers several shortcomings. The most apparent is its
time complexity: inherently, the one-hot encoding procedure used in \Rappor leads to a time complexity of $\Theta(\ab\ns)$, with an extra linear dependence on the alphabet size $\ab$, which is far from the ``gold standard'' of $O(\ns)$ complexity.

A more time-efficient procedure is obtained using \HR. In fact, we describe an algorithm for testing uniformity based on \HR that has the same sample complexity as the one based on \Rappor described above, but is much more time-efficient. 
\begin{theorem}\label{thm:ub-hadamard:fourier}
For $\eps\in (0,1]$,~\cref{alg:hadamard:uniformity} based on  $\eps$-LDP \HR can test whether a distribution is uniform or $\dst$-far from uniform using \[O\Paren{\frac{\ab^{3/2}}{\dst^2\eps^2}}\] samples. Moreover, the algorithm runs in time near-linear in the number of samples.
\end{theorem}

\begin{algorithm}[ht]
\begin{algorithmic}[1]
\Require Privacy parameter $\eps>0$, distance parameter $\dst\in(0,1)$, $\ns$ samples
\State Set
\[
\alpha_{H} \gets \frac{e^{\eps}-1}{e^{\eps}+1}\, \qquad K \gets    2^{\Ceil{\log(k+1)}}
\]
\State Apply the \HR (with parameters $\eps$, $K$) to the $\ns$ samples to obtain $\ns$ independent samples in $[K]$ \Comment{Time $O(\log\ab)$ per user}
\State Invoke the testing algorithm $\textsc{Test}$-$\lp[2]$ of~\cref{theo:cdvv:l2} on these $\ns$ samples, with parameters 
\[
b\gets \frac{1+\alpha_H}{\sqrt{K}},\qquad \dst' \gets \frac{2\alpha_H\dst}{\ab K} 
\] 
and $\q^\ast$ being the explicit distribution from~\cref{thm:ub-hadamard:parseval} \Comment{Time $O(\ns\log \ab+\ns\log \ns)$} 
\If{$\textsc{Test}$-$\lp[2]$ accepts}
  \State \Return \textsf{uniform}
\Else
  \State \Return \textsf{not uniform}
\EndIf
\end{algorithmic}
\caption{Locally Private Uniformity Testing using Hadamard Response}\label{alg:hadamard:uniformity}
\end{algorithm}

To describe the intuition behind this algorithm, suppose we feed
inputs from an input distribution $\p\in\distribs{[\ab]}$ to the 
more general mechanism in~\cref{sec:general-private}, whose output then follows some
induced distribution $\q\in\distribs{[K]}$. A natural hope is that
whenever $\p$ is uniform (over $[\ab]$), then $\q$ is uniform (over $[K]$), too; and
that conversely if $\p$ is not uniform, then $\q$ is neither, and
 that the distance to uniformity is preserved. This is not 
exactly what we will obtain. However, we can get something close to it in the next result, which suffices for our purpose.\footnote{To see that our desired 
statement cannot hold as stated above, note
that for $\p=\uniform$,~\eqref{eq:expression:output} implies  
$\q(z_1) = \frac{1+\alpha_{H}}{K}$, since $\abs{D_{z_1}}=\ab$ as
the first column of $H_K$ is the all-one vector. Thus the squared
$\lp[2]$ distance of $\q$ to uniform is at least
$\alpha_{H}^2/K$.}
\begin{theorem}
\label{thm:ub-hadamard:parseval}
Let $\eps\in(0,1]$, $K=O(\ab)$ be a power of $2$, and denote by $\q$
the output distribution over $[K]$. Then, we have
\begin{align}
\normtwo{\q-\q^\ast}^2= \frac{\alpha_{H}^2}{K}\cdot \normtwo{\p-\uniform}^2
\asymp \frac{\eps^2}{\ab}\normtwo{\p-\uniform}^2\,,
\end{align}
where $\alpha_{H} \eqdef \frac{e^{\eps}-1}{e^{\eps}+1}$, and
$\q^\ast\in\distribs{[K]}$ is an explicit distribution, efficiently
computable and independent of $\p$, with $\normtwo{\q^\ast}\leq
(1+\alpha_{H})/\sqrt{K}$. Moreover, $\q^\ast$ can be sampled in time
$O(\log K)$.
\end{theorem}
Thus,  when $\p=\uniform$, we get $\q=\q^\ast$. Otherwise when $\totalvardist{\p}{\uniform}>\dst$,
then
\begin{align}
\normtwo{\q-\q^\ast}^2 > \frac{4\alpha_{H}^2\dst^2}{\ab K} =\Theta\!\Paren{\frac{\eps^2}{\ab^2}\dst^2}.
\end{align}

The observation above suggests that if we can estimate the $\lp[2]$ distance between $\q$ and $\q^\ast$, we can get our desired uniformity test.  We facilitate this by invoking the result below, which follows from the
$\lp[2]$-distance estimation algorithm of~\cite[Proposition 3.1]{ChanDVV14}, combined with an observation from~\cite[Lemma 2.3]{DiakonikolasK16}:\footnote{\cite{ChanDVV14} require that
$b\geq \max(\normtwo{\p}, \normtwo{\q})$; \cite{DiakonikolasK16} shows 
how to relax this requirement to $b\geq \min(\normtwo{\p}, \normtwo{\q})$.}
\begin{theorem}[{Adapted from~\cite[Proposition 3.1]{ChanDVV14}}]\label{theo:cdvv:l2}
For two unknown distributions $\p,\q\in\distribs{[\ab]}$, there exists an algorithm $\textsc{Test}$-$\lp[2]$ that distinguishes with probability at least $2/3$ between the cases
$\normtwo{\p-\q} \leq \dst/2$ and $\normtwo{\p-\q}
> \dst$ by observing $O(\min(\normtwo{\p},\normtwo{\q})/\dst^2)$ samples from each.   Moreover, this algorithm runs in time near-linear in the number of samples.
\end{theorem}

We apply the algorithm of~\cref{theo:cdvv:l2} to our case by generating desired number of samples from $\q^\ast$, which can simply be obtained by passing samples from the uniform distribution via \HR, and using them along with the samples observed from $\q$ at the output of \HR. We need to distinguish between the cases $\q=\q^\ast$ and
$\normtwo{\q-\q^\ast}>\dst'/\sqrt{K}$, which by the previous result can be done using
$O(\normtwo{\q^\ast}K/\dst'^2)$ samples where $\dst'\eqdef 2\alpha_{H}\dst/\sqrt{\ab}$.  
Substituting  $K=O(\ab)$ and $\normtwo{\q^\ast}=O(1/\sqrt{K})$, the number of samples we need is 
\[
  O\!\Paren{\frac{1}{\sqrt{K}}\cdot
  K \cdot \Paren{\frac{\sqrt{\ab}}{\dst\eps}}^2} =
  O\!\Paren{\frac{\ab^{3/2}}{\dst^2\eps^2}},
\]
which is our claimed sample complexity. 

 The time complexity follows from the efficiency of
Hadamard encoding (see~\cite[Section 4.1]{AcharyaSZ18}), which
allows each player to generate their private sample in time $O(\log
K)=O(\log \ab)$, and to send only $O(\log \ab)$ bits.\footnote{This is significantly better than the $O(\ab)$ time and communication per player
of~\cref{alg:rappor:chisquare}.} After this, running the
$\textsc{Test}$-$\lp[2]$ algorithm takes time $O(\ns\log
K+\ns\log\ns)$, the first term being the time required to generate
$\ns$ samples from $\q^\ast$. Thus, to conclude the proof
of~\cref{thm:ub-hadamard:fourier}, it only remains to
establish~\cref{thm:ub-hadamard:parseval} -- which we do next.

\begin{proofof}{\cref{thm:ub-hadamard:parseval}}
For any $z\in[K]$, let $D_z\subseteq[\ab]$ be the set of symbols $x$
such that $z\in C_x$. Then, from~\cref{eq:general:asz} (recalling that
$s=K/2$) we get
\begin{align}
  \q(z) &= \sum_{x\in D_z} \mech(z\mid x)\p(x) + \sum_{x\in
   D_z} \mech(z\mid x)\p(x)
\nonumber
\\ 
&= \frac{1}{s}\Paren{ \p(D_z) \frac{e^\eps}{e^\eps+1}
   + \p(D_z^c) \frac{1}{e^\eps+1} } 
   +\frac{1}{K}\cdot\frac{e^\eps-1}{e^\eps+1}\cdot\Paren{\p(D_z)-\p(D_z^c)}
\nonumber
\\ 
  &=\frac{1}{K} + \frac{\alpha_{H}}{K}\Paren{2\p(D_z)-1}\,.
\label{eq:expression:output}
\end{align}
Define $\q^\ast\in\distribs{[K]}$ as
\begin{equation}\label{eq:def:qast}
\q^\ast(z) \eqdef \frac{1}{K}
    + \frac{\alpha_{H}}{K}\Paren{\frac{2\abs{D_z}}{\ab}-1}, \quad     \forall z\in[K],
\end{equation}
so that $\q^\ast(z) \in [(1-\alpha_{H})/K, (1+\alpha_{H})/K]$ for
every $z$ and $\normtwo{\q^\ast} \leq (1+\alpha_{H})/\sqrt{K}$.
From~\cref{eq:expression:output,eq:def:qast}, we get
\begin{equation}\label{eq:normtwo:q:expression}
    \normtwo{\q-\q^\ast}^2 = \sum_{z\in[K]}\Paren{\q(z) -\q^\ast(z)}^2
    = \frac{4\alpha_{H}^2}{K^2}\sum_{z\in[K]} \left( \p(D_z)-\frac{\abs{D_z}}{\ab} \right)^2.
\end{equation}

 Note that we may view a probability distribution $\p\in\distribs{[\ab]}$ as a function $\textbf{\p}\colon \{0,1\}^\ab\to \RR$ with
\[
 \textbf{\p}(s) = \p(S)
    = \sum_{x\in S} \p(x), \quad \forall s\in\{0,1\}^\ab,
\]
where we identify $s\in\{0,1\}^\ab$ with the subset
$S=\setOfSuchThat{x\in[\ab]}{ s_x = 1}\subseteq [\ab]$ and use the two
notations interchangeably. Also, from~\cref{fact:hadamard:fourier} and the
definition of the
$C_x$ as sets encoded by the rows of the matrix $H_K$, we  have
that
\[
 \chi_{\phi(x)}(z) =
    (H_K)_{\phi(x),z} = 2\indic{z\in C_x} - 1, \quad     \forall\, z\in[K], x\in[\ab],
\]
whereby
\begin{equation*}
   \sum_{x\in[\ab]} \left( \p(x)
    - \frac{1}{\ab} \right)\chi_{\phi(x)}(z) =
    2\left( \p(D_z)-\frac{\abs{D_z}}{\ab} \right), \quad \forall\, z\in[K].
\end{equation*}
Now, consider the function $g\colon [K]\to \RR$
defined by $g(z) =
2\left( \p(D_z)-\abs{D_z}/\ab \right)$,\footnote{Recall that $K$ is a
power of two, so $[K]$ can be identified to $\{0,1\}^{\log K}$.}.
Using the previous equation, we can view $g$ alternatively as
\[
g(z)
    = \sum_{T\in\phi([\ab])}\left( \p( \phi^{-1}(T) )
    - \frac{1}{\ab} \right)\chi_T(z), \quad     \forall\, z\in[K].
\]
Therefore, 
\[
\hat{g}(T) = \sum_{x\in[\ab]}\left( \p(x) -
1/\ab \right)\indic{T=\phi(x)}, \quad \forall\, T\in[K],
\]  
which by Parseval's theorem (\cref{theo:parseval}) gives
\begin{align}
\nonumber
  \frac{4}{K}\sum_{z\in[K]} \left( \p(D_z)-\frac{\abs{D_z}}{\ab} \right)^2
  = \norm{g}^2 = \sum_{T\in [K]} \hat{g}(T)^2
  = \sum_{x\in[\ab]} \left( \p(x) - \frac{1}{\ab} \right)^2
  = \normtwo{\p-\uniform}^2\,.
\end{align}
The identity above, together with~\eqref{eq:normtwo:q:expression}, yields
\[
    \normtwo{\q-\q^\ast}^2
    = \frac{\alpha_{H}^2}{K}\normtwo{\p-\uniform}^2\,.
\]
The claimed result then follows from the fact that
$\alpha_{H}=\frac{\eps}{2}+o(\eps)$.
\end{proofof}

\section{Optimal Locally Private Uniformity Testing}\label{sec:optimal-uniformity}
In the foregoing treatment, we saw that existing (private-coin)
mechanisms such as \Rappor and \HR can perform
uniformity testing using $O(\ab^{3/2}/(\dst^2 \eps^2))$ samples at
best.  In this section, we describe our public-coin mechanism,
\Raptor,\footnote{Which stands for \emph{Randomized Aggregated Private
    Testing Optimal Response}.} and use it to design an algorithm for
testing uniformity that requires only $O(\ab/(\dst^2 \eps^2))$ samples
and constant communication\footnote{In fact, we only need $1$-bit per sample if we allow asymmetric implementation.} per sample.

Our algorithm builds upon the warmup
algorithm of~\cref{lem:binary-testing}, which allows us to perform uniformity testing for $\ab=2$
using $O(1/(\dst^2\eps^2))$ samples. Specifically, we use public
randomness to reduce the uniformity testing problem for an arbitrary
$\ab$ to that for $\ab =2$, albeit with $\dst$ replaced with
$\dst/\sqrt{\ab}$; and then apply the warmup algorithm.

To enable the aforementioned reduction, we need to show that the
probabilities of a randomly generated set differ appropriately under
the uniform distribution and a distribution that is $\dst$ far from
uniform in total variation distance. To accomplish this, we prove a
more general result which might be of independent interest.  We say
that random variables 
$X_1, X_2, \dots, X_\ab$ are \emph{$4$-symmetric} if
$\expect{X_{i_1}X_{i_2}X_{i_3}X_{i_4}}$ depends only on the number of
times each element appears in the multiset $\{i_1, i_2, i_3,
i_4\}$.\footnote{That is, if $\expect{X_{i_1}X_{i_2}X_{i_3}X_{i_4}}$ does not depend on the actual values of $i_1,i_2,i_3,i_4$ (for which there are $\ab^4$ possibilities) but only on the quantities $\indic{i_a=i_b}$, for $1\leq a< b\leq 4$ (for which there are $2^{\binom{4}{2}}$ possibilities).}{} The following result constitutes a concentration bound for $Z
= \sum_{i\in [\ab]}\delta_iX_i$ for a probability perturbation
$\delta$.

\begin{theorem}[Probability perturbation concentration]\label{theorem:random:subset}
Consider a vector $\delta$ such that $\sum_{i\in [\ab]}\delta_i =
0$. Let random variables $X_1, \dots, X_\ab$ be  $4$-symmetric and $Z = \sum_{i\in [\ab]}\delta_iX_i$. 
Then,  for every $\alpha \in (0,1/4)$,
\[
\probaOf{  
\bigg(\expect{X_1^2} - \expect{X_1X_2}\bigg) -  \sqrt{\frac{38 \alpha}{1-2\alpha}\expect{X_1^4}}
\leq \frac{Z^2}{\normtwo{\delta}^2} \leq  \frac{1}{1-2\alpha} 
\bigg(\expect{X_1^2} - \expect{X_1X_2}\bigg)
} \geq \alpha. 
\]
\end{theorem}
\noindent The proof requires a careful evaluation of the second and the fourth
moments of $Z$ and is deferred to~\cref{app:concentration}. 
As a corollary, we obtain the result below, which is at the core of our
reduction argument. 
\begin{corollary}\label{corollary:random:subset}
Consider a distribution $\p\in \distribs{[\ab]}$ such that
 $\totalvardist{\p}{\uniform}>\dst$. For a random subset $S$ of $[\ab]$
   distributed uniformly over all subsets of $[\ab]$ of
 cardinality $\ab/2$, it holds that
\[ 
\probaOf{\abs{p(S)-\frac{1}{2}}>\frac{\dst}{\sqrt{5\ab}}}>\frac{1}{477}\,.  
\]
\end{corollary}
\begin{proof}
Let $Y_1, \dots, Y_\ab$ be independent random bits, and let
$X_1, \dots, X_\ab$ be obtained by conditioning $(Y_1, \dots, Y_\ab)$
on the event $\sum_{i\in [\ab]}X_i = \ab/2$.  Consider the random set
$S$ defined as
\[
S = \setOfSuchThat{i\in[\ab] }{ X_i =1}.
\]
Letting $\delta \eqdef \p - \uniform$, we have
\[
\p(S) - \uniform(S) = \p(S) - \frac 1 2 = \sum_{i\in[\ab]}\delta_iX_i.
\]
Note that $\expect{X_1^2} = \expect{X_1^4} =1/2$. Also,
\begin{align*}
\expect{X_1X_2} &= \frac 12 \probaCond{X_2=1 }{ X_1=1}
\\
&= \frac 12 \probaCond{Y_2=1 }{ Y_1=1, \sum_{i=1}^\ab Y_i =\frac \ab 2
}
\\
&= \frac 12 \probaCond{Y_2=1 }{ \sum_{i=2}^\ab Y_i =\left(\frac \ab 2
-1\right)}
\\
&= \frac 12 \expectCond{Y_2 }{ \sum_{i=2}^\ab Y_i =\left(\frac \ab 2
-1\right)}
\end{align*}
which by symmetry yields
\[
\expect{X_1X_2} = \frac 12 \expectCond{Y_j }{ \sum_{i=2}^\ab Y_i =\left(\frac \ab 2
-1\right)}, \quad \forall\, 2\leq j\leq \ab.
\]
Taking the average of the right-side, we get
\[
\expect{X_1X_2}
= \frac 1{2(\ab-1)} \expectCond{\sum_{j=2}^\ab Y_j }{ \sum_{i=2}^\ab
Y_i =\left(\frac \ab 2 -1\right)} =\frac{(\ab
-2)}{4(\ab-1)} \leq \frac {1}{4}.
\]
Finally, note that
\[
\normtwo{\delta}^2 \geq \frac{1}{\ab} \normone{\delta}^2 \geq \frac{4\dst^2}{\ab},
\]
so that, applying~\cref{theorem:random:subset} to $Z
= \sum_{i=1}^\ab \delta_i X_i$ with $\alpha \eqdef 1/477$ we get
\[
\probaOf{Z^2 >\frac{\dst^2}{5\ab}} \geq \frac{1}{477},
\]
which completes the proof.
\end{proof}
Armed with this result, we can divide our LDP testing problem into two parts: A public-coin $\eps$-LDP mechanism releases $1$-bit per sample to the curator, and the curator applies a test to the received bits to accomplish uniformity testing. This specific mechanism suggested by the previous corollary is our \Raptor (see~\cref{alg:raptor} for a description). While in this paper we have only considered its use for testing uniformity and independence, since it provides locally private $1$-bit outputs that, in essence, preserve the $\lp[2]$ distance of the underlying distribution from any other fixed one, we can foresee many other use-cases for \Raptor and pose it as a standalone mechanism of independent interest.

Recall that in \Raptor the curator and the users pick a random subset $S$ of size $\ab/2$ from
their shared randomness, and each user sends the indicator function
that its input lies in this set $S$ using $\eps$-LDP \RaR. This is precisely the $1$-bit information from samples required to enable the estimator of~\cref{lem:binary-testing}.
Note that when the underlying distribution $\p$ is
uniform, the probability $p(S)$ of user bit being $1$ is exactly
1/2. Also, 
by~\cref{corollary:random:subset} when $\p$ is $\dst$-far from uniform
we have $p(S) = 1/2\pm\Omega(\dst/\sqrt{\ab})$ with a constant
probability (over the choice of $S$); by repeating the protocol a
constant number of times,\footnote{To preserve the symmetry of our mechanism, we note that this can be done ``in parallel'' at each user. That is, each user considers the same $T = \Theta(1)$ many random subsets, and sends their corresponding $T$ privatized (with parameter $\eps' = \eps/T$) indicator bits to the curator.}
we can ensure that with high constant probability
at least one of the choices of $S$ will indeed have this
property. Therefore, we obtain an instance of the uniformity testing
problem for $\ab=2$, namely the problem of 
privately distinguishing a $\bernoulli{1/2}$ from
$\bernoulli{1/2\pm \frac{c_1\dst}{\sqrt{\ab}}}$.
Thus, when we apply \Raptor to the samples, the curator gets the $1$-bit updates required by~\cref{lem:binary-testing} to which it can apply the estimator prescribed in~\cref{lem:binary-testing} to solve the underlying uniformity testing instance for $\ab=2$ using
\[
O\Paren{{\frac{\ab}{\dst^2}\frac{(e^\eps+1)^2}{(e^\eps-1)^2}}}
\]
samples. Since we used $\eps$-LDP \RaR to send each bit, \Raptor, too, is $\eps$-LDP and thereby so is our overall uniformity test. 

We summarize the overall algorithm and its performance below.
\begin{theorem}
\label{thm:ub-public}
For $\eps\in (0,1]$,~\cref{alg:optimal:uniformity} based on $\eps$-LDP \Raptor can test whether a distribution is uniform or $\dst$-far from uniform using \[O\Paren{\frac{\ab}{\dst^2\eps^2}}\] samples.
\end{theorem}
\begin{algorithm}[ht]
\begin{algorithmic}[1]
\Require Privacy parameter $\eps>0$, distance parameter $\dst\in(0,1)$, $\ns = mT$ samples
\State Set
\[
    c \gets \frac{1}{477}\, \qquad \delta \gets \frac{c}{2(1+c)}, \qquad \dst' \gets \frac{\dst}{\sqrt{5\ab}}, \qquad T = \Theta(1)
\]
\For{$t$ from $1$ to $T$} \Comment{In parallel}
  \State Generate uniformly at random a subset $S_t$ of $[\ab]$ of cardinality $\ab/2$

\State Apply \Raptor using $S_t$ to each sample in the mini-batch of $m$ samples
 
  \State Use the estimator in the proof of~\cref{lem:binary-testing} to test with probability of failure $\delta$ 
{\Statex \hspace{2em}if 
      $\p(S_t) = 1/2$ (\textsf{unbiased}) or $\abs{\p(S_t) - 1/2}> \dst'$ (\textsf{biased})}
\EndFor
\State Let $\tau$ denote the fraction of the $T$ outcomes that returned $\textsf{unbiased}$
\If{$\tau > 1- (\delta + \frac{c}{4})$}
  \State \Return \textsf{uniform}
\Else
  \State \Return \textsf{not uniform}
\EndIf
\end{algorithmic}
\caption{Locally Private Uniformity Testing using \Raptor}\label{alg:optimal:uniformity}
\end{algorithm}
\begin{proofof}{\cref{thm:ub-public}}
  The proof of correctness follows the foregoing outline, which we describe in more detail. Let $c\eqdef 1/477$ be the constant from~\cref{corollary:random:subset}, and let $\delta \eqdef \frac{c}{2(1+c)}$, and set $\dst' \eqdef \frac{\dst}{\sqrt{5\ab}}$. By a standard amplification argument,\footnote{Namely, letting the server divide the received samples into $O(\log(1/\delta))$ disjoint batches, and running the private estimation procedure of~\cref{lem:binary-testing} independently $T$ times, before outputting the majority vote.} one can amplify the success probability of the private estimation procedure of~\cref{lem:binary-testing} to $1-\delta$, using a total of $O(\log(1/\delta)/({\dst'}^2\eps^2))=O(1/({\dst'}^2\eps^2))$ samples (to achieve privacy $\eps$ and accuracy $\dst'$).
  
Consider the $t$-th test from~\cref{alg:optimal:uniformity} (where $1\leq t\leq T$), and let $b_t$ be the indicator that the bias estimation outputs \textsf{unbiased}. If $\p$ is uniform, then by the above we have $\probaOf{ b_t = 1 } \geq 1- \delta$ (where the probability is over the choice of the random subset $S_t$, and the randomness of the bias estimation). However, if  $\p$ is $\dst$-far from uniform, by~\cref{corollary:random:subset} it it the case that $\probaOf{ b_t = 1 } \leq (1-c) + c\delta = 1- (\delta+\frac{c}{2})$. Therefore, for a sufficiently large constant in the choice of $T = \Theta(1/c^2) = \Theta(1)$, a Chernoff bound argument ensures that we can distinguish between these two cases with probability at least $2/3$.
\end{proofof}

\section{Proofs of Optimality: Lower Bounds for Uniformity Testing}\label{sec:lower-uniformity}
We now establish that the public-coin testing algorithm in the
previous section has optimal sample-complexity for {\em any} LDP
uniformity testing algorithm.  Furthermore, we establish lower bounds
on the sample complexity for any LDP testing algorithm using \Rappor
or \HR, showing that the tests we proposed using these mechanisms are sample optimal (up to constant factors) in their class. 

\subsection{Lower bound for public-coin mechanisms}\label{sec:public:coin:lb}
We first show that any uniformity testing algorithm that uses data from an $\eps$-LDP public-coin mechanism (which includes private-coin mechanisms) requires at least $\Omega\Paren{\ab/(\dst^2\eps^2)}$ samples. 
\begin{theorem}
\label{thm:lb-public}
For $\eps\in (0,1]$, any $\eps$-LDP public-coin mechanism for uniformity testing must use
\[\Omega\Paren{\frac{\ab}{\dst^2\eps^2}}\] samples.
\end{theorem}
\begin{proof}
Our lower bound relies on analyzing the standard ``Paninski construction''~\cite{Paninski08}, which we briefly recall. Assuming without loss of generality that $\ab$ is even, we partition the domain in $\ab/2$ consecutive pairs $(2i-1,2i)$. For a given parameter $\dst\in(0,1/2]$, the family of ``\no-instances'' is the collection of $2^{\ab/2}$ distributions $(\p_{\theta})_{\theta\in\{-1,+1\}^{\ab/2}}$ where 
\[
    \p_{\theta}(2i-1) = \frac{1+2\theta_i\dst}{\ab}, \qquad \p_{\theta}(2i) = \frac{1-2\theta_i\dst}{\ab}, \qquad i\in[\ab/2]\,.
\]
Note that every such $\p_\theta$ is a total variation exactly $\dst$ from the uniform distribution on $[\ab]$.

Our starting point will be the proof of the public-coin lower bound of~\cite[Theorem 6.1]{ACT:18} for uniformity testing in a (non-private) distributed setting. Note that the proof in~\cite{ACT:18} proceeds by noting that once we restrict our attention to the hypothesis testing problem implied by Paninski's construction, we can derandomize and find a deterministic protocol that outperforms the public-coin protocol. Therefore, it suffices to bound the performance of deterministic protocols. However, in our current application, relaxing to deterministic protocols will get rid of local privacy constraints and will not lead to useful bounds. Instead, we note in similar vein as the proof in~\cite{ACT:18} that we can derandomize public randomness and find a private-coin $\eps$-LDP protocol that achieves the same performance for Paninski's construction as the public-coin protocol we start with. Therefore, it suffices to restrict our attention to private-coin protocols. 

Let $\mech$ be an arbitrary  $\eps$-LDP private-coin mechanism for uniformity testing. For $1\leq j\leq \ns$ and $\theta,\theta'\in\{-1, +1\}^{\ab/2}$, define $H_j(\theta, \theta')$ as
\begin{align*}
    H_j(\theta, \theta') 
    &\eqdef \frac{\dst^2}{\ab} \sum_{m}\sum_{i_1,i_2\in [\ab/2]}\theta_{i_1}\theta'_{i_2}    \frac{ \left(
      \mech_j(m \mid 2i_1-1)-\mech_j(m \mid 2i_1) \right) \left(
      \mech_j(m \mid 2i_2-1)-\mech_j(m \mid 2i_2) \right)}{ \sum_{i=1}^{\ab} 
      \mech_j(m \mid i)}\,.
\end{align*}
where $m\in(\{0,1\}^\ast)^\ns$ denotes the tuple of outputs from the $\ns$ users. Let 
\[
\bar{\mech}_j(m) \eqdef \frac{1}{\ab} \sum_{i=1}^{\ab} \mech_j(m \mid i)
\]
$i.e.$, $\bar{\mech}_j(m)$ is the probability of user $j$ outputting $m$ when the input distribution is uniform. Let $\delta_{i,j}^m$ be such that 
\[
\mech_j(m \mid 2i-1)-\mech_j(m \mid 2i)= \bar{\mech}_j(m) \delta_{i,j}^m.
\]
Then by the conditions for LDP, we must have 
\begin{align*}
\abs{ \delta_{i,j}^m }\leq e^\eps-1\,.
\end{align*}
Furthermore, we can rewrite 
\begin{align*}
    H_j(\theta, \theta') 
    &= \frac{\dst^2}{\ab^2} \sum_{m}\sum_{i_1,i_2\in [\ab/2]}\theta_{i_1}\theta'_{i_2}{\bar{\mech}_j(m)\delta_{i_1,j}^m\delta_{i_2,j}^m}
    = \sum_{i_1,i_2\in [\ab/2]}\theta_{i_1}\theta'_{i_2}\Paren{\frac{\dst^2}{\ab^2}\sum_{m}\bar{\mech}_j(m)\delta_{i_1,j}^m\delta_{i_2,j}^m}\\
    &= \frac{\dst^2}{\ab^2}\theta^T \matH_j \theta',
\end{align*}
where $\matH_j$ is an $[\ab/2]\times [\ab/2]$ matrix with $(i_1, i_2)$th entry equal to
\begin{equation*}
\matH_j(i_1, i_2) \eqdef {\sum_{m}\bar{\mech}_j(m)\delta_{i_1,j}^m\delta_{i_2,j}^m}\,.
\end{equation*}
By using that $\sum_m \bar{\mech}_j(m) = 1$, we further get that
\begin{equation*}
\abs{\matH_j(i_1, i_2)} \leq \Paren{e^\eps-1}^2\,.
\end{equation*}
For a given distribution $\p\in\distribs{[\ab]}$, denote by $\mathcal{\mech}(\p)\in\distribs{(\{0,1\}^\ast)^\ns}$ the product distribution over $m$ (the tuple of $\ns$ messages) when each user gets an independent sample from $\p$. 
In~\cite{ACT:18}, it is shown that, with $\ns$ users, the $\chi^2$ distance between the distributions of (i)~the output of the mechanism under the Paninski mixture, $\mathcal{\mech}^{\no} \eqdef \frac{1}{2^{\ab/2}}\sum_{\theta} \mathcal{\mech}(\p_\theta)$, and (ii)~the output of the mechanism under the uniform distribution $\mathcal{\mech}^\yes \eqdef \mathcal{\mech}(\uniform)$, is bounded by
\begin{align*}
\EE_{\theta, \theta'}\Paren{\exp\Paren{\sum_{j=1}^\ns  H_j(\theta, \theta')}}-1
=\EE_{\theta, \theta'}\Paren{\exp\Paren{\frac{\dst^2}{\ab^2}\theta^T (\sum_{j=1}^\ns \matH_j)\theta'}}-1\,.
\end{align*}
We will also rely on the following technical claim:
\begin{claim}[{\cite[Claim 6.10]{ACT:18}}]\label{claim:uniformity:lb:technical:transportation}
Consider random vectors $\theta, \theta'\in\{-1, 1\}^{\ab/2}$ with each $\theta_i$ and $\theta'_i$ distributed uniformly over $\{-1, 1\}$, independent of each other and independent for different $i$'s. Then, for any symmetric matrix $H$,
\[
\ln \EE_{\theta \theta'}{ e^{\lambda\theta^T H\theta'}} \leq \lambda^2 \norm{H}_F^2, \quad \forall\, \lambda>0\,.
\]
\end{claim}
Using this claim, and choosing $\lambda \eqdef \frac{\dst^2}{\ab^2}$ we obtain the following upper bound on the distance between the distributions of the output of the mechanisms in the two cases:
\begin{align*}
\totalvardist{\mathcal{\mech}^\no}{\mathcal{\mech}^\yes} 
&\leq \exp\Paren{\frac{\dst^4}{\ab^4}\norm{\sum_{j=1}^\ns \matH_j}_F^2}-1
\leq \exp\Paren{\frac{\dst^4}{\ab^4}n\Paren{\sum_{j=1}^\ns\norm{\matH_j}_F^2}}-1\\
&\leq \exp\Paren{\frac{\dst^4 n^2}{\ab^4}\Paren{(k/2)^2 \Paren{e^\eps-1}^4}}-1
\leq \exp\Paren{\frac{\dst^4 n^2}{4\ab^2}\Paren{e^\eps-1}^4}-1\,.
\end{align*}
This implies the claimed lower bound by a standard application of Le Cam's two-point method (as e.g. detailed by Pollard~\cite{Pollard:2003}), as one must have $\ns = \Omega(\ab/(\eps^2\dst^2))$ for the RHS to be $\Omega(1)$. 
\end{proof}
 \subsection{Lower bound for \Rappor}\label{ssec:rappor:lb}
In this section, we prove a lower bound for {\em any} uniformity testing mechanism that uses \Rappor, not only the algorithm from~\cref{sec:rappor} (\cref{thm:ub-rappor}). In fact, the next result shows that that algorithm requires the least number of samples (up to constant factors) among all mechanisms based on the output \Rappor, even those allowing public-coin protocols in their post-processing stage. 
\begin{theorem}
\label{thm:lb-rappor}
In the high-privacy regime, any $\eps$-LDP mechanism for uniformity testing 
that uses \Rappor for reporting user data must use
\[
\Omega\Paren{\frac{\ab^{3/2}}{\dst^2\eps^2}}
\]
samples.
\end{theorem}
\begin{proof}
We once again take recourse to Pollard's recipe and proceed as in the proof of~\cref{thm:lb-public}. Denote by $\mech\colon [\ab] \to \{0,1\}^\ab$ the channel from the input to the output of \Rappor. Letting for conciseness $q_\eps \eqdef e^\eps/(1+e^\eps)$, we first observe that for any observation $i\in [\ab]$ one has
\begin{align*}
\mech(m\mid i) &= \left(m_i q_\eps + (1-m_i)(1-q_\eps)\right)\cdot q_\eps^{\abs{\setOfSuchThat{ \ell }{ \ell \neq i, m_\ell =0}}} (1-q_\eps)^{\abs{\setOfSuchThat{\ell}{\ell \neq i, m_\ell =1}}}
\\
 &= m_i q_\eps^{m(0)+1}(1-q_\eps)^{m(1)-1} 
+ (1-m_i) q_\eps^{m(0)-1}(1-q_\eps)^{m(1)+1},
\end{align*}
where $m(0)$ and $m(1)$, respectively, denote the number of $0$'s and $1$'s in $m$. As was seen in the previous proof of lower bound, for our purpose, we need to evaluate  $\sum_{i \in [\ab]}\mech(m\mid i)$ and $\mech(m\mid i) - \mech(m\mid i')$. For the former quantity, we have 
\begin{align*}
\sum_{i\in [\ab]}\mech(m\mid i) &= q_\eps^{m(0)}(1-q_\eps)^{m(1)} \left(m(1)  \frac{q_\eps}{1-q_\eps} + m(0) \frac{1-q_\eps}{q_\eps}\right).
\end{align*}
For the latter, we have 
\begin{align*}
\mech(m\mid i) - \mech(m\mid i')&= q_\eps^{m(0)}(1-q_\eps)^{m(1)} (m_i - m_{i'})\left(\frac{q_\eps}{1-q_\eps}
- \frac{1-q_\eps}{q_\eps}\right).
\end{align*}
We are now in a position to prove the lower bound. We are considering protocols where each sample $X_j$ is reported to the center using \Rappor, and so, each sample is reported using the same channel $W$ described above. Therefore, the $\ab\times \ab$ matrix $\matH_j$ used in the previous proof does not depend on $j$ and satisfies 
\begin{align*}
    \matH(i_1, i_2) \eqdef \matH_j(i_1, i_2) &= \ab\sum_{m}\frac{ \left(
      W(m|2i_1-1)-W(m|2i_1) \right) \left(
      W(m|2i_2-1)-W(m|2i_2) \right)}{ \sum_{i=1}^{\ab} 
      \mech(m\mid i)}.
\end{align*}
It follows that 
\begin{align*}
    \matH(i_1, i_2) &= \ab \cdot \frac{(e^{2\eps} -1 )^2}{e^{\eps}}
\cdot \sum_{m} q_\eps^{m(0)}(1-q_\eps)^{m(1)}\cdot 
\frac{(m_{2i_1 - 1} - m_{2i_1})  (m_{2i_2 - 1} - m_{2i_2})}{m(1)e^{2\eps}+ m(0)}.
\end{align*}
The key observation that facilitates our bound is that for $i_1\neq i_2$, the sum on the right-side is $0$. Indeed, consider the set of messages of fixed {\em type}, namely those with $m(0)$ and $m(1)$ fixed. Note that in any such set, only messages with $m_{2i_1}\neq m_{2i_1-1}$ and $m_{2i_2}\neq m_{2i_2-1}$ contribute to the sum. Furthermore, for any fixed $m_{2i_2}\neq m_{2i_2-1}$, the contributions corresponding to $m_{2i_1}=1, m_{2i_1-1}=0$ and $m_{2i_1}=0, m_{2i_1-1}=1$ negate each other when $i_1\neq i_2$, whereby the overall sum is $0$. Thus, we have
\begin{align*}
\matH(i_1, i_2)&=\indic{i_1=i_2}
\frac{\ab(e^{2\eps} -1 )^2}{e^{\eps}}
\cdot \sum_{m} q_\eps^{m(0)}(1-q_\eps)^{m(1)}\cdot 
\frac{|m_{2i_1 - 1} - m_{2i_1}|^2 }{m(1)e^{2\eps}+ m(0)}
\\
&\leq 
\indic{i_1=i_2}
\frac{(e^{2\eps} -1 )^2}{e^{\eps}}
\cdot \sum_{m} q_\eps^{m(0)}(1-q_\eps)^{m(1)}
\\
&=
\indic{i_1=i_2}
\frac{(e^{2\eps} -1 )^2}{e^{\eps}},
\end{align*}
where the inequality holds since $m(1)e^{2\eps} + m(0)\geq \ab$ for every $\eps\geq 0$. It follows that 
\[
\norm{\matH}_F^2\leq \ab \cdot \frac{(e^{2\eps} -1 )^4}{e^{2\eps}},
\]
whereby
\begin{align*}
\totalvardist{\mathcal{W}^\no}{\mathcal{W}^\yes}
&\leq \exp\Paren{\frac{\dst^4\ns^2}{\ab^4}\norm{\matH}_F^2}-1
\leq \exp\Paren{\frac{\dst^4\ns^2}{\ab^4}\Paren{\ab\frac{(e^{2\eps}-1)^4}{e^{2\eps}}}}-1\\
&\le \exp\Paren{\frac{\dst^4 \ns^2}{e^{2\eps}\ab^3}\Paren{e^{2\eps}-1}^4}-1,
\end{align*}
which can only be $\Omega (1)$ if 
$\ns = \Omega\Paren{\frac{\ab^{3/2}}{\dst^2\Paren{e^\eps-1}^2}}$, establishing the result.
\end{proof}

\subsection{Lower bound for Hadamard Response}
Finally, we establish the analogue of~\cref{thm:lb-rappor} for any mechanism based on Hadamard Response.
\begin{theorem}
\label{thm:lb-har}
In the high-privacy regime, any $\eps$-LDP mechanism for uniformity testing 
that uses \HR for reporting user data must use
\[
\Omega\Paren{\frac{\ab^{3/2}}{\dst^2\eps^2}}
\]
samples.
\end{theorem}
\begin{proof}
The proof follows the same outline as the proof of~\cref{thm:lb-rappor}~--~we show once again that the 
matrix $\matH$ corresponding to using \HR for reporting each users data is a diagonal  matrix. Specifically, 
considering the set of messages as $\{0,\dots, K-1\}$ and the inputs as $\{0,\dots \ab-1\}$ for convenience,
the matrix $\matH$ in the proof of lower bound is given by 
\begin{align}
    \matH(i_1, i_2) &= \ab\sum_{m=0}^{K-1}\frac{ \left(
      \mech(m\mid 2i_1)-\mech(m\mid 2i_1+1) \right) \left(
      \mech(m\mid 2i_2)-\mech(m\mid 2i_2+1) \right)}{ \sum_{i=0}^{\ab-1} \mech(m\mid i)}, \quad 0\leq i_1, i_2\leq \frac{\ab-1}{2}\,,
\nonumber
\end{align}
where $\mech(m\mid i)$ denotes the probability that \HR  outputs $m$ when the input is $i$
and is given by
\begin{align*}
\mech(m\mid i) &= 
\frac{2}
{K(e^\eps+1)}\Paren{ e^{\eps} \indic{m\in C_i}+ \indic{m\notin C_i} }
=\frac{2}{K(e^\eps+1)}\Paren{ (e^{\eps}-1) \indic{m\in C_i}+ 1 }\,.
\end{align*}
With a slight abuse of notation, we use $m$ and $i$ interchangeably to denote their values and the binary vectors corresponding to binary representation of those values. Further, let $\dotprod{x}{y}  = \oplus\, x_i y_i$ denote the standard (parity) inner product for vectors over $\mathbb{F}_2$. With this convention, for \HR  we have
\[
\indic{m\in C_i} = \indic{\dotprod{m}{i} =0}.
\]
Note that for $m=0$ and every $i$, $\indic{m\in C_i} = 1$, which implies that $\mech(0\mid i)$ is the same for all $i$. It follows that the term corresponding to $m=0$ in the expression for $\matH(i_1, i_2)$ is $0$.
Moreover, using $\sum_{i=0}^{\ab-1} \indic{\dotprod{m}{i} =0}=K/2$ for $m\neq 0$, we obtain
\begin{align*}
\sum_{i=1}^{\ab}\mech(m\mid i) &=  \frac{2}{K(e^\eps+1)}\Paren{ \Big(e^{\eps}-1\Big)\sum_{i=0}^{\ab-1} \indic{\dotprod{m}{i} =0}+ \ab }
\\
&=\frac{2}{K(e^\eps+1)}\Paren{ \Big(e^{\eps}-1\Big)\frac K2+ \ab }
\\
&= c(\eps, \ab, K)\cdot\frac{\ab}{K}\,,
\end{align*}
where $c(\eps, \ab, K) =2((e^{\eps}-1) K/2+ \ab )/(\ab(e^\eps+1))$; by assumption, $\ab\leq K \leq 2\ab$, which implies that $c(\eps, \ab, K) \in [1, 2]$. Therefore, for every $0\leq i_1, i_2\leq (\ab-1)/2$,
\begin{align}
    \matH(i_1, i_2) &= c(\eps, \ab, K)\cdot\frac{4(e^\eps-1)^2}{K(e^\eps+1)^2}\sum_{m=1}^{K-1}
 \left(\indic{\dotprod{m}{2i_1} =0} - \indic{\dotprod{m}{2i_1+1} =0}\right)\cdot
\nonumber
\\
&\hspace{6cm} \left(\indic{\dotprod{m}{2i_2} =0} - \indic{\dotprod{m}{2i_2+1} =0}\right).
\label{e:H-terms}
\end{align}
We claim  that $\matH(i_1, i_2)=0$ for $i_1\neq i_2$. Indeed, a case analysis yields 
\begin{align*}
\left(\indic{\dotprod{m}{2i_1} =0} - \indic{\dotprod{m}{2i_1+1} =0}\right) &\left(\indic{\dotprod{m}{2i_2} =0} - \indic{\dotprod{m}{2i_2+1} =0}\right)\\
&= \left(\indic{\dotprod{m}{2i_1} = \dotprod{m}{2i_2}} - \indic{\dotprod{m}{2i_1} \neq \dotprod{m}{2i_2}}\right)\indic{\dotprod{m}{1} \neq 0}\,,
\end{align*}
where the condition $\indic{\dotprod{m}{1} \neq 0}$ indicates that the expression is nonzero only when $m$ is odd. Thus, the summands in \eqref{e:H-terms} can be restricted to odd $m$, and further, each summand equals 
\[
\indic{\dotprod{m}{2i_1} = \dotprod{m}{2i_2}} - \indic{\dotprod{m}{2i_1} \neq \dotprod{m}{2i_2}}
= \indic{\dotprod{m}{2(i_1\oplus i_2)}=0} - \indic{\dotprod{m}{2(i_1\oplus i_2)} =1}\,.
\]
We can simplify the expression on the right-side by noting that each odd $m$ has the binary form $(b,1)$ and 
$\dotprod{(b,1)}{2i} = \dotprod{b}{i}$. Hence, 
\begin{align*}
\sum_{m=1}^{K-1} \left(\indic{\dotprod{m}{2i_1} =0} - \indic{\dotprod{m}{2i_1+1} =0}\right)&\left(\indic{\dotprod{m}{2i_2} =0} - \indic{\dotprod{m}{2i_2+1} =0}\right)
\\
&=\sum_{m=0}^{(K-1)/2}
\left(\indic{\dotprod{m}{(i_1\oplus i_2)}=0} - \indic{\dotprod{m}{(i_1\oplus i_2)}=1}\right)
\\
&= \frac{K-1}{2}\indic{i_1=i_2}.
\end{align*}
since for any nonzero vector $j$, $|\setOfSuchThat{m}{ \dotprod{m}{j} \rangle =0}| = |\setOfSuchThat{m}{ \dotprod{m}{j} \rangle =1 }|$\,. 
In summary, we have 
\[
 \matH(i_1, i_2) =c(\eps, \ab, K)\cdot \frac{4(e^\eps-1)^2(K-1)}{K(e^\eps+1)^2} \indic{i_1=i_2},
\]
whereby using $c(\eps, \ab, K)\in [1,2]$ we get
\[
 \norm{\matH}_F^2 = \sum_{i_1, i_2} |\matH(i_1, i_2)|^2 = O\left(\ab\frac{(e^\eps-1)^4}{(e^\eps+1)^4}\right).
\]
The proof is completed in the same manner as the proof of~\cref{thm:lb-rappor}. 
\end{proof}

\section{Independence Testing}\label{sec:independence}
In this section, we treat independence testing.  We begin
in~\cref{sec:independence-hadamard} with an independence testing
mechanism based on \HR that does not require public
randomness and which achieves significantly improved sample
complexity over the state-of-the-art (in dependence on the alphabet
size). However, we do not have matching lower bounds for its
performance.

Then, in~\cref{sec:independence-optimal}, we describe and analyze an
optimal procedure that uses public randomness, akin to the optimal
uniformity testing mechanism of~\cref{sec:optimal-uniformity}. 

\subsection{A mechanism based on Hadamard Response and private $\chi^2$ learning}\label{sec:independence-hadamard}
We present a symmetric, private-coin LDP
mechanism for testing independence of distributions over
$[\ab]\times[\ab]$ (although, as we note
in~\cref{rk:independence:otherdomains}, our mechanism can be easily
extended to handle a more general setting).

\begin{theorem}\label{thm:privatecoin:ub-independence}
For $\eps\in (0,1]$,~\cref{alg:private:independence} based on a
  symmetric private-coin $\eps$-LDP mechanism can test whether a
  distribution over $[\ab]\times[\ab]$ is a product distribution or
  $\dst$-far from product using
\[O\Paren{\frac{\ab^{3}}{\dst^2\eps^4}}\] samples.
\end{theorem}

The rest of this section is dedicated to the proof of~\cref{thm:privatecoin:ub-independence}. The argument will follow the ``testing-by-hybrid-learning'' approach of~\cite{AcharyaDK15},  modified suitably for the local differential privacy setting. Specifically, instead of learning and testing the underlying user data distributions, we simply do this in the Fourier domain for the distributions seen at the output of \HR; details follow.

Denote by $\mathcal{H}$ the \HR mapping from $[\ab]$ to $[K]$, where $K=O(\ab)$, and $\alpha_{H} \eqdef \frac{e^\eps-1}{e^\eps+1}$ as in~\cref{sec:hadamard}. For any probability distribution $\p\in\distribs{[\ab]\times[\ab]}$, we define $\mathcal{T}(\p)\in\distribs{[K]\times[K]}$ as the distribution of $(Z_1,Z_2)$ obtained by the process below:
\begin{enumerate}[(1)]
  \item Draw $(X_1,X_2)$ from $\p$;
  \item apply \HR independently to $X_1$ and $X_2$ to obtain $Z_1$ and $Z_2$.
\end{enumerate}
It is immediate to see that if $\p$ is a product distribution with
marginals $\p_1$ and $\p_2$, then
\[
    \mathcal{T}(\p) = \mathcal{H}(\p_1)\otimes\mathcal{H}(\p_2)\,.
\] 
We build our test on observations $(Z_1, Z_2)$ from each
user, with distribution $\mathcal{T}(\p)$. Our proposed test uses
these samples. It builds on several components that we will describe
later; for ease of presentation, we summarize the overall algorithm in~\cref{alg:private:independence}. 
\begin{algorithm}[ht]
\begin{algorithmic}[1]
\Require Privacy parameter $\eps>0$, distance parameter $\dst\in(0,1)$
\State Set
\[
\alpha_{H} \gets \frac{e^{\eps}-1}{e^{\eps}+1}\,, \qquad K \gets 2^{\Ceil{\log(k+1)}}
\]
\State Each of the $\ns=2\ns_1+ \ns_2$ users, given their data
$(X_1,X_2)$, applies \HR independently to $X_1$ and $X_2$ and sends
the outcomes to the curator
\State Use the first $2\ns_1$ samples from $\cT(\p)$ to obtain $\ns_1$ samples from $\cT(\p_1\otimes\p_2)$

\State Apply the algorithm of~\cref{theo:chisquare:learning} to the
$\ns_1=O(\ab^3/\alpha_H^4\dst^2)$ samples from $\cT(\p_1\otimes \p_2)$
obtained in the previous step to learn a distribution
$\q\in\distribs{[K]\times[K]}$ such that $\min_z \q(z) \geq 1/(50K^2)$ and,
with probability at least $4/5$,  
   \[
        \chi^2(\mathcal{T}(\p_1\otimes\p_2),\q)\leq \alpha_{H}^4\dst^2/\ab^2
    \] \label{algo:independence:learning}
\State Apply the algorithm of~\cref{theo:chisquare:adk} to the remaining  $\ns_2=O(\sqrt{K^2}/{\dst'}^2) = O(\ab^3/(\alpha_{H}^4\dst^2))$ samples from $\mathcal{T}(\p)$ to distinguish, with probability at least $9/10$, between
   \[
        \chi^2(\mathcal{T}(\p),\q)\leq {\dst'}^2/2,\qquad \text{and}\qquad \chi^2(\mathcal{T}(\p),\q) > {\dst'}^2
    \] where ${\dst'}^2 \eqdef 2\alpha_{H}^4\dst^2/\ab^2$. \label{algo:independence:testing}
\If{$\chi^2(\mathcal{T}(\p),Q)\leq {\dst'}^2/2$}
  \State \Return \textsf{independent}
\Else
  \State \Return \textsf{not independent}
\EndIf
\end{algorithmic}
\caption{Locally Private Independence Testing}\label{alg:private:independence}
\end{algorithm}
For~\cref{algo:independence:testing}, we rely on a result of~\cite{AcharyaDK15}, modified slightly for our purposes:\footnote{This statement differs slightly from that in~\cite{AcharyaDK15},
  but can be seen to follow from their analysis. Indeed, the
  difference only impacts the analysis of the variance of their
  estimator, which now goes through because of our assumption on
  $\min_{x\in[\ab]}\q(x)$.}
\begin{theorem}[{\cite[Theorem 1]{AcharyaDK15}}]\label{theo:chisquare:adk}
Given the explicit description of a
distribution $\q\in\distribs{[\ab]}$ such that $\min_{x\in[\ab]}\q(x)
\geq \frac{1}{50\ab}$ and samples from an unknown distribution
$\p\in\distribs{[\ab]}$,  there exists an efficient algorithm with
that can distinguish with probability at least $9/10$ between the cases
$\chi^2(\p,\q)<\frac{\dst^2}{2}$ and $\chi^2(\p,\q) \geq
\dst^2$ using $O(\sqrt{\ab}/\dst^2)$ samples.
\end{theorem}

In addition to~\cref{theo:chisquare:adk}, our proposed algorithm builds on~\cref{theo:chisquare:learning}
which we will describe and prove below. But before we
prove this result, we note that this algorithm can be
seen to satisfy all the properties claimed
in~\cref{thm:privatecoin:ub-independence}. Indeed, it requires
$O(\ab^3/(\alpha_{H}^4\dst^2))$; its privacy is immediate since the
observations at the curator are obtained by passing user data via
\HR. The mechanism is clearly symmetric, as each user sends the output
of \HR (applied independently to both marginal of their data) to the
curator -- it is only at the curator that these privatized outputs are
used and combined to generate samples from $\mathcal{T}(\p)$,
$\mathcal{T}(\p_1\otimes\p_2)$, $\mathcal{T}(\p_1)$, or
$\mathcal{T}(\p_2)$. 

As for the correctness, it will follow
from~\cref{theo:chisquare:learning,theo:chisquare:adk} (ensuring that the algorithm is overall correct with probability at least
$7/10>2/3$), once the following structural property is established:
For $\p$ that is $\dst$-far from any product distribution has
$\chi^2(\mathcal{T}(\p),\q) > {\dst'}^2$. Formally, we show the following:
\begin{theorem}\label{coro:fourier:independence:structural}
  Let $\p\in\distribs{[\ab]\times[\ab]}$ with marginals $\p_1,\p_2\in\distribs{[\ab]}$, and $\q\in\distribs{[\ab]\times[\ab]}$ be a product distribution such that $\chi^2(\mathcal{T}(\p_1\otimes\p_2),\mathcal{T}(\q)) \leq \frac{\alpha_{H}^4\dst^2}{\ab^2}$. (i)~If $\p$ is a product distribution, then $\chi^2( \mathcal{T}(\p), \mathcal{T}(\q) ) \leq \frac{\alpha_{H}^4\dst^2}{\ab^2}$. (ii)~if $\p$ is $\dst$-far from being a product distribution, then $\chi^2( \mathcal{T}(\p), \mathcal{T}(\q) ) > \frac{2\alpha_{H}^4\dst^2}{\ab^2}$.
\end{theorem}

\noindent It only remains to establish the structural property above
and $\chi^2$ learning algorithm~\cref{theo:chisquare:learning}.

\paragraph{Proof of structural result~\cref{coro:fourier:independence:structural}.} We prove that (i)~if $\p$ is independent, then $\chi^2(\mathcal{T}(\p_1\otimes\p_2),\q)$ will be small, while (ii)~if $\p$ is far from independent then $\chi^2(\mathcal{T}(\p_1\otimes\p_2),\q)$ must be noticeably larger. The key technical component is the next lemma. 

\begin{lemma}\label{lemma:ub-hadamard:parseval:general}
    Let $\p,\q\in\distribs{[\ab]\times[\ab]}$ be two distributions, with marginals $\p_1,\p_2$ and $\q_1,\q_2$, respectively. Then,
      \[
            \normtwo{\mathcal{T}(\p)-\mathcal{T}(\q)}^2 = \frac{\alpha_{H}^4}{K^2}\normtwo{\p-\q}^2 + \frac{\alpha_{H}^2}{K^2}\left( \normtwo{\p_1-\q_1}^2 + \normtwo{\p_2-\q_2}^2 \right)\,.
      \]
      In particular, if $\totalvardist{\p}{\q} > \dst$, then $\normtwo{\mathcal{T}(\p)-\mathcal{T}(\q)} > \frac{2\alpha_{H}^2\dst}{K\ab}$.
\end{lemma}
\begin{proof}
The proof is similar to that of~\cref{thm:ub-hadamard:parseval}. Note that it follows in the manner of~\eqref{eq:expression:output} that for every $(z_1,z_2)\in[K]\times[K]$,
  \begin{align*}
      \mathcal{T}(\p)&(z_1,z_2)
\\
&= \sum_{(x_1,x_2)\in [\ab]\times[\ab]} \mech(z_1 \mid x_1)\mech(z_2 \mid x_2)\p(x_1,x_2) \\
      &= \frac{4}{K^2(e^\eps+1)^2} \sum_{(x_1,x_2)\in [\ab]\times[\ab]} \p(x_1,x_2) \Paren{ (e^\eps-1) \indic{z_1\in C_{x_1}} + 1 }\Paren{ (e^\eps-1) \indic{z_2\in C_{x_2}} + 1 } \\
      &= \frac{1}{K^2}\sum_{(x_1,x_2)\in [\ab]\times[\ab]} \p(x_1,x_2) \Paren{ \alpha_{H}\chi_{\phi(x_1)}(z_1) + 1 }\Paren{ \alpha_{H}\chi_{\phi(x_2)}(z_2) + 1 } \\
      &= \frac{\alpha_{H}^2}{K^2}\sum_{x_1,x_2} \p(x_1,x_2)\chi_{\phi(x_1)}(z_1)\chi_{\phi(x_2)}(z_2) 
        + \frac{\alpha_{H}}{K^2}\Big( \sum_{x_1} \p_1(x_1)\chi_{\phi(x_1)}(z_1) + \sum_{x_2} \p_2(x_2)\chi_{\phi(x_2)}(z_2) \Big)
        + \frac{1}{K^2},
  \end{align*}
  where by the second-to-last identity above gives $\mathcal{T}(\p)(z_1,z_2)\in [1-\alpha_{H}^2,1+\alpha_{H}^2]\cdot \frac{1}{K^2}$ for every $(z_1,z_2)$. As an analogous expression holds for  $\mathcal{T}(\q)(z_1,z_2)$, setting $g\eqdef \mathcal{T}(\p) - \mathcal{T}(\q)$ we have
  \begin{align*}
      g(z_1,z_2) &= \frac{\alpha_{H}^2}{K^2}\sum_{(x_1,x_2)\in [\ab]\times[\ab]} \Paren{\p(x_1,x_2) - \p_1(x_1)\p_2(x_2) } \chi_{\phi(x_1)}(z_1)\chi_{\phi(x_2)}(z_2) \\
      &\qquad+ \frac{\alpha_{H}}{K^2}\Big( \sum_{x_1\in [\ab]} (\p_1(x_1)-\q_1(x_1)) \chi_{\phi(x_1)}(z_1) + \sum_{x_2\in [\ab]} (\p_2(x_2)-\q_2(x_2)) \chi_{\phi(x_2)}(z_2) \Big),
  \end{align*}  
for every $(z_1,z_2)\in [K]\times[K]$. Now, as
in~\cref{thm:ub-hadamard:parseval}, but looking at the corresponding
characters for the Hadamard transform from
$[\ab]\times[\ab]$ to $[K]\times[K]$,\footnote{There are $(K+1)^2$
  characters for $[K]\times[K]$ of the form $\chi_T(z_1)\chi_S(z_2)$,
  $\chi_T(z_1)$, and $\chi_T(z_2)$ for $S,T\subseteq [K]$, along with
  the constant character. Note that the constant character will not
  appear in the proof of~\cref{lemma:ub-hadamard:parseval:general}, as
  we consider the transform of the difference of two functions,
  canceling the constant term.}{} we get
  \begin{align*}
  \hat{g}(T) &= \frac{\alpha_{H}^2}{K^2}\sum_{(x_1,x_2)\in[\ab]\times[\ab]}\Paren{\p(x_1,x_2) - \p_1(x_1)\p_2(x_2) }\indic{T=(\phi(x_1),\phi(x_2))} \\
  &\qquad+ \frac{\alpha_{H}}{K^2}\Big( \sum_{x_1\in[\ab]}\Paren{\p_1(x_1)-\q_1(x_1)}\indic{T=\phi(x_1)} + \sum_{x_2\in[\ab]}\Paren{\p_2(x_2)-\q_2(x_2)}\indic{T=\phi(x_2)} \Big)\,.
  \end{align*}
  By Parseval's theorem (\cref{theo:parseval}),
  \begin{align*}
   \norm{g}^2 &=   \frac{1}{K^2}\sum_{(z_1,z_2)\in[K]\times[K]}\Paren{ \mathcal{T}(\p)(z_1,z_2) - \mathcal{T}(\q)(z_1,z_2) }^2     
\\
&= \sum_{T\in[K]\times[K]} \hat{g}(T)^2 
\\
      &= \frac{\alpha_{H}^4}{K^4}\sum_{(x_1,x_2)\in[\ab]\times[\ab]}\Paren{\p(x_1,x_2) - \q(x_1,x_2) }^2     
\\
&\hspace{2cm}
+ \frac{\alpha_{H}^2}{K^4}\Big( \sum_{x_1}\Paren{\p_1(x_1)-\q_1(x_1)}^2 + \sum_{x_2}\Paren{\p_2(x_2)-\q_2(x_2)}^2 \Big),
  \end{align*}
  so that
  \[
    \normtwo{\mathcal{T}(\p)-\mathcal{T}(\q)}^2 = \frac{\alpha_{H}^4}{K^2}\normtwo{\p-\q}^2 + \frac{\alpha_{H}^2}{K^2}\left( \normtwo{\p_1-\q_1}^2 + \normtwo{\p_2-\q_2}^2 \right)
    \geq \frac{\alpha_{H}^4}{K^2}\normtwo{\p-\q}^2\,,
  \]
  as claimed.
\end{proof}

\begin{proofof}{\cref{coro:fourier:independence:structural}}
    The first statement is obvious, as then $\mathcal{T}(\p) = \mathcal{T}(\p_1\otimes\p_2)$. Turning to the second, assume that $\p$ is $\dst$-far from being a product distribution, so that in particular $\totalvardist{\p}{\q} > \dst$. This implies by~\cref{lemma:ub-hadamard:parseval:general} that
   $
    \normtwo{\mathcal{T}(\p) - \mathcal{T}(\q)}^2 > (4\alpha_{H}^4\dst^2)/(K^2\ab^2)
    $ 
which  along with $\norminf{\mathcal{T}(\q)} \leq \frac{1+\alpha_{H}^2}{K^2} \leq \frac{2}{K^2}$ yields
    \[
        \chi^2( \mathcal{T}(\p), \mathcal{T}(\q) ) = \sum_{z\in[K]\times[K]}\frac{ (\mathcal{T}(\p)(z) - \mathcal{T}(\q)(z))^2 }{ \mathcal{T}(\q)(z)}
        \geq \frac{K^2}{2}\normtwo{\mathcal{T}(\p) - \mathcal{T}(\q)}^2
        > \frac{2\alpha_{H}^4\dst^2}{\ab^2},
    \]
    as claimed.
\end{proofof}

\paragraph{Learning in $\chi^2$ distance in the Hadamard domain.}
Next, we establish the correctness of Step~\ref{algo:independence:learning}. That is, we show that by 
leveraging the product structure one can (privately) learn
$\mathcal{T}(\p_1\otimes\p_2)$ to the desired $\chi^2$ accuracy with the
number of samples scaling as $\ab^3$. To do so, will rely on the
following result in~\cite{KamathOPS15} on learning in $\chi^2$ distance: 
\begin{lemma}[{\cite[Lemma 4]{KamathOPS15}}]\label{lemma:learning:chisquare:kops}
The Laplace (add-1) estimator can learn $\ab$-ary distributions to $\chi^2$ distance $\dst^2$, with probability $9/10$, using $O(\ab/\dst^2)$ samples. 
\end{lemma}
\noindent The next corollary ensues.
\begin{corollary}
    There is an efficient estimator to learn product distributions over $[\ab]\times[\ab]$ to $\chi^2$ distance $\dst^2$, with probability $4/5$, using $O(\ab/\dst^2)$ samples.
\end{corollary}
\begin{proof}
    Let $\p=\p_1\otimes\p_2\in\distribs{[\ab]\times[\ab]}$ be a
    product distribution, and $\tilde{\p}_1$, $\tilde{\p}_2$ be the
    hypotheses obtained by using the estimator
    of~\cref{lemma:learning:chisquare:kops} independently on the two
    marginals of $\p$, with distance parameter $\dst^2/3$. We claim
    that $\tilde{\p}\eqdef \tilde{\p}_1\otimes\tilde{\p}_2$ can serve
    as our desired estimate. Indeed, 
    by a union bound, with probability at least $4/5$ it is the case
    that $\chi^2(\p_1,\tilde{\p}_1) \leq \dst^2/3$ and
    $\chi^2(\p_2,\tilde{\p_2}) \leq \dst^2/3$. When this happens,
    \begin{align*}
        \chi^2(\p,\tilde{\p}) 
        &= -1 + \sum_{(x_1,x_2)\in[\ab]\times[\ab]} \frac{\p_1(x_1)^2\p_2(x_2)^2}{\tilde{\p}_1(x_1)\tilde{\p}_2(x_2)}
\\   
    & = -1 + \sum_{x_1\in[\ab]} \frac{\p_1(x_1)^2}{\tilde{\p}_1(x_1)}\sum_{x_2\in[\ab]} \frac{\p_2(x_2)^2}{\tilde{\p}_2(x_2)} \\
        &= -1 + (\chi^2(\p_1,\tilde{\p}_1)+1)(\chi^2(\p_2,\tilde{\p}_2)+1)
\\     
  & = \chi^2(\p_1,\tilde{\p}_1)\chi^2(\p_2,\tilde{\p_2})+\chi^2(\p_1,\tilde{\p}_1)+\chi^2(\p_2,\tilde{\p_2}) \\
        &\leq \frac{\dst^4}{9} + \frac{2}{3}\dst^2  < \dst^2,
    \end{align*}
    concluding the proof.
\end{proof}

\noindent As a further corollary, we finally obtain the desired
algorithm for LDP $\chi^2$-learning.
\begin{corollary}\label{theo:chisquare:learning}
    There exists an algorithm based on a private-coin, symmetric $\eps$-LDP mechanism that
    learns $\mathcal{T}(\p_1\otimes\p_2)$ to
    $\chi^2$ distance $\dst'^2 \eqdef \alpha_{H}^4\dst^2/\ab^2$, with
    probability $4/5$, using $O(\ab^3/(\alpha_{H}^4\dst^2))$
    samples. Moreover, the estimate $\q\in\distribs{[K]\times[K]}$
    obtained by this algorithm is a product distribution with $\min_{z\in[K]\times[K]}\q(z) \geq 1/(50K^2)$.
\end{corollary}
\noindent The last point follows from the fact that
$\mathcal{T}(\p_1\otimes\p_2)(z) \geq 1/(2K^2)$ for every
$z\in[K]\times[K]$. Hence, if $\q(z) < 1/(50K^2)$ for some $z$, then
$\chi^2(\mathcal{T}(\p_1\otimes\p_2), \q) \gg 1$. 

\begin{remark}\label{rk:independence:otherdomains}
  To conclude this section, we note that a straightforward
  generalization of~\cref{alg:private:independence,theo:chisquare:learning,theo:chisquare:adk,coro:fourier:independence:structural} to the case $[\ab_1]\times[\ab_2]$ leads to a symmetric private-coin
  $\eps$-LDP mechanism to test whether a distribution over
  $[\ab_1]\times[\ab_2]$ is a product distribution vs. $\dst$-far from
  product with sample complexity  
$
O\Paren{(\ab_1 \ab_2)(\ab_1 + \ab_2 + \sqrt{\ab_1 \ab_2})/(\dst^2\eps^4)}
$.
\end{remark}

\subsection{Optimal independence testing using public-coin mechanisms}\label{sec:independence-optimal}
We proceed, as for uniformity testing, by reducing independence
testing for arbitrary $\ab$ to that for $\ab=2$; an algorithm for solving the latter problem is given as a
warmup in~\cref{cor:binary-ind}.  We show that if
$\totalvardist{\p}{\p_1\otimes\p_2}\geq \dst$, uniformly random sets
$S_1, S_2\subseteq[\ab]$ of cardinality $\ab/2$ satisfy
\begin{equation}\label{eq:s1s2}
  \abs{ \p(S_1, S_2)- \p_1(S_1)\p_2(S_2) } = \Omega(\dst/\ab),
\end{equation}
with constant probability. Therefore, we can perform our independence
test by repeating the mechanism of~\cref{cor:binary-ind} 
$O(1)$ times, each for independently generated $S_1, S_2$ applied to 
$O(\ab^2/(\dst^2\eps^2))$ samples. Indeed, the claim above guarantees
that, with high constant probability, when
$\totalvardist{\p}{\p_1\otimes\p_2}\geq \dst$, one of the $O(1)$
repetitions will produce sets $S_1,S_2$ that satisfy~\eqref{eq:s1s2}.
On the other hand, clearly when $\p=\p_1\otimes\p_2$, we have
$\p(S_1\times S_2)=\p_1(S_1)\p_2(S_2)$ for all $S_1,
S_2\subseteq[\ab]$. Thus, the mechanism described
in~\cref{cor:binary-ind} will allow us to perform independence testing
using $O(\ab^2/(\dst^2\eps^2))$ samples by obtaining first the
estimates $\tilde{\p}$, $\tilde{\p}_1$, and $\tilde{\p}_2$,
respectively, of $\p(S_1\times S_2)$, $\p_1(S_1)$, and $\p_2(S_2)$,
and then comparing $\abs{\tilde{\p}-\tilde{\p}_1\tilde{\p}_2}$ with a
suitable threshold.\footnote{As was the case for our uniformity
  testing algorithm, to preserve the symmetry of our mechanism we can
  ask that these $O(1)$ repetitions be done ``in parallel'' at each
  user; further, for each of these repetitions every user sends three
  privatized bits to the central server, corresponding to the
  indicators of $S_1\times[\ab]$, $[\ab]\times S_2$, and $S_1\times
  S_2$.}\medskip

Hence, we can (as we did for uniformity testing) divide the LDP
testing problem into two parts: A public-coin $\eps$-LDP mechanism 
releases 3 bits per sample to the curator, and then the curator
applies a test to the received bits to perform independence testing on
the reduced domain $\{0,1\}\times\{0,1\}$. The specific mechanism
underlying the first part will be \Raptor (specifically, a bivariate
variant of \Raptor given in~\cref{algo:raptor:bivariate}); the second part
relies on the estimator of~\cref{cor:binary-ind}. As
in~\cref{sec:optimal-uniformity}, we can boost the probability of
success to $2/3$ by performing the above two-part test a constant
number of times and using the median trick.

\begin{algorithm}[H]
\begin{algorithmic}[1]
\State The curator and the users sample two independent and uniformly
random subsets $S_1,S_2$ of $[\ab]$ of cardinality $\ab/2$.  \State
Each user computes the three bit indicators
  \[
      B_{1,i}=\indic{X_{1,i} \in S_1},\qquad B_{2,i}=\indic{X_{2,i}
        \in S_1},\qquad B_{i}=\indic{(X_{1,i},X_{2,i}) \in S_1\times
        S_2}
  \] and sends them using $\RaR$, $i.e.$, flips each of them independently with probability
$1/(1+e^{\eps/3})$ and sends the outcome to the
  curator. \Comment{Parameter $\eps/3$ to obtain $\eps$-LDP of the
    joint 3 bits.}
\end{algorithmic}
\caption{\label{algo:raptor:bivariate}The \Raptor mechanism, bivariate
  version}
\end{algorithm}

\smallskip 

It only remains to prove the claim~\eqref{eq:s1s2}. This requires the
following (somewhat technical) extension
of~\cref{theorem:random:subset}; for simplicity, we provide a less
general version that addresses only a specific choice of random
variables $X$ and $Y$.
\begin{theorem}[Joint probability perturbation concentration]\label{theorem:random:product:subsets}
Consider a matrix $\delta\in\RR^{\ab\times \ab}$ such that, for every
$i_0,j_0\in[\ab]$, $\sum_{j\in [\ab]}\delta_{i_0,j}=\sum_{i\in
  [\ab]}\delta_{i,j_0} = 0$. Let random variables $X=(X_1, \dots,
X_\ab)$ and $Y=(Y_1, \dots, Y_\ab)$ be independent and uniformly
distributed over $\ab$-length binary sequences of weight $\ab/2$.
Define $Z = \sum_{(i,j)\in [\ab]\times[\ab]}\delta_{ij}X_iY_j$.  Then,
there exist constants $c_1,c_2,\rho>0$ such that
\[
\probaOf{\frac{Z^2}{ \norm{\delta}_F^2}\in [c_1,c_2]}\geq \rho.
\]
\end{theorem}
\noindent We provide the details of the proof
of~\cref{theorem:random:product:subsets}
in~\cref{app:concentration:bivariate}. In particular, choosing
$\delta_{ij} \eqdef (\p(i,j) - \p_1(i)\p_2(j))$, we obtain the desired
result as a corollary:
\begin{corollary}
Consider $\p\in\distribs{[\ab]\times[\ab]}$ with marginals $\p_1,\p_2$
such that $\totalvardist{\p}{\p_1\otimes\p_2}\geq \dst$.  For randomly
chosen subsets $S_1$ and $S_2$, generated uniformly and independently
over all subsets of $[\ab]$ of cardinality $[\ab/2]$, there exist
positive constants $c$ and $\rho$ such that
\[
\probaOf{|\p(S_1\times S_2)-\p_1(S_1)\p_2(S_2)|\geq c \frac
  {\totalvardist{\p}{\p_1\otimes\p_2}}{\ab}}\geq \rho\,.
\]
\end{corollary}
\begin{proof}
Setting $\delta_{ij} = (\p(i,j) - \p_1(i)\p_2(j))$ in
\cref{theorem:random:product:subsets}, note that
\[
Z = \sum_{i,j}X_iY_j(\p(i,j)-\p_1(i)\p_2(j))= \p(S_1\times
S_2)-\p_1(S_1)\p_2(S_2)
\]
and that
\[
 \norm{\delta}_F^2 = \sum_{i,j\in [\ab]}\delta_{i,j}^2 \geq \frac{
   \Big(\sum_{i,j\in [\ab]}\abs{\delta_{i,j}}\Big)^2}{\ab^2}\geq
 \frac{4\dst^2}{\ab^2}.
\]
Thus, the claim follows from~\cref{theorem:random:product:subsets}.
\end{proof}
Finally, we show that the sample requirement for our mechanism is
optimal for $\eps \in (0,1]$. The proof is similar to that
  of~\cref{thm:lb-public}, since the uniform distribution on
  $[\ab]\times[\ab]$ is also a product distribution. The only caveat
  is that we need to ensure that each perturbed distribution is at
  total variation distance at least $\dst$ from \emph{every} product
  distribution, not only the uniform one.  Fortunately, we can get this
 using the following simple fact:
\begin{fact}
  Assume $\p$ is $\dst$-close to some product distribution
  $\q\in\distribs{[\ab]\times[\ab]}$. Then, $\p$ is $\dst$-close to the
  product distribution induced by its own marginals, $i.e.$,
  $\totalvardist{\p}{\p_1\otimes\p_2} \leq 3\dst$.
\end{fact}
\noindent Consequently, to show that the perturbed distribution is
$\dst$-far from independent it is enough to prove it is $(3\dst)$-far
from the product of its marginals, which in turn is
immediate. This implies that locally private independence
testing is information-theoretically at least as hard as locally
private uniformity testing over $[\ab]\times[\ab]$ ($i.e.$, over
alphabet size $\ab^2$), yielding the
$\Omega\Paren{{\ab^2}/{(\dst^2\eps^2)}}$ sample lower bound. 
In summary, combining the upper and lower bounds we have shown the
following result. 
\begin{theorem}
\label{thm:ub-public-indpe}
There exists a symmetric, public-coin $\eps$-LDP mechanism to test
whether a distribution over $[\ab]\times[\ab]$ is a product
distribution vs. $\dst$-far from product using
\[O\Paren{\frac{\ab^2}{\dst^2\eps^2}}\]
samples. Furthermore, any $\eps$-LDP mechanism for
testing independence in this regime must use
$\Omega\Paren{\frac{\ab^2}{\dst^2\eps^2}}$ samples.
\end{theorem}

\paragraph{Acknowledgments.} We thank Jon Ullman for bringing to our attention the relation between symmetric and asymmetric schemes,~\cref{fact:asymmetric:advantage}, and outlining its proof.

\bibliographystyle{alpha}
\bibliography{references}
\appendix
\section{Proof of~\cref{theorem:random:subset}}\label{app:concentration}

\begin{theorem}[Probability perturbation concentration, restated]
Consider a vector $\delta$ such that $\sum_{i\in [\ab]}\delta_i =
0$. Let random variables $X_1, \dots, X_\ab$ be  $4$-symmetric and $Z = \sum_{i\in [\ab]}\delta_iX_i$. 
Then, for every $\alpha \in (0,1/4)$,
\[
\probaOf{  
\bigg(\expect{X_1^2} - \expect{X_1X_2}\bigg) -  \sqrt{\frac{38 \alpha}{1-2\alpha}\expect{X_1^4}}
\leq \frac{Z^2}{\normtwo{\delta}^2} \leq  \frac{1}{1-2\alpha} 
\bigg(\expect{X_1^2} - \expect{X_1X_2}\bigg)
} \geq \alpha. 
\]
\end{theorem}
\begin{proof}
Since $X_1, X_2,\dots, X_\ab$ are $4$-symmetric, the expectations $\expect{X_aX_bX_cX_d}$ depends only on the number of times each distinct element appears in the multiset $\{a,b,c,d\}$. For ease of notation, we replace the highest frequency element in $\{a,b,c,d\}$ with $1$, second highest with $2$, and so on, to obtain a representation $S$ and denote $m_{S} \eqdef \expect{\prod_{i\in S}X_i}$. For instance, $\expect{X_a^2} = m_{\{1,1\}}$ and $\expect{X_aX_b}=m_{\{1,2\}}$ for distinct $a,b$. With this notation at our disposal, we are ready to proceed with the proof.

Note first that 
\begin{equation}
\expect{Z} =  \sum_{i\in[\ab]}\delta_i \expect{X_i}  = m_{\{1\}} \sum_{i\in [\ab]} \delta_i =0.
\end{equation}
Moreover, for the variance of $Z$, we have 
\begin{align}
\variance{Z} = \expect{Z^2} &=  \sum_{i_1,i_2 \in [\ab]}\delta_{i_1}\delta_{i_2} \expect{X_{i_1}X_{i_2}}  \notag\\
&=  \sum_{i_1 \in [\ab]}\delta_{i_1}^2\expect{X_{i_1}^2} + \sum_{i_1\neq
i_2}\delta_{i_1}\delta_{i_2}\expect{X_{i_1}X_{i_2}} \notag\\
&=  m_{\{1,1\}}\normtwo{\delta}^2 +m_{\{1,2\}} \sum_{i\neq
j}\delta_{i_1}\delta_{i_2} \notag\\
&=  m_{\{1,1\}}\normtwo{\delta}^2
+m_{\{1,2\}} \left(\Big(\sum_{i\in[\ab]}\delta_i\Big)^2
- \sum_{i\in [\ab]}\delta_i^2\right) \notag\\
&= (m_{\{1,1\}}- m_{\{1,2\}})\normtwo{\delta}^2 \notag\\
&= \bigg(\expect{X_1^2}-\expect{X_1X_2}\bigg)\normtwo{\delta}^2 \label{eq:paley:1}\,,
\end{align}
where we used $\sum_{i\in[\ab]}\delta_i = 0$ in the previous
identity. It follows from Chebyshev's
inequality that 
\[
\probaOf{Z^2\leq \frac 1{1-2\alpha} \left(\expect{X_1^2}-\expect{X_1X_2}\right)\normtwo{\delta}^2} \geq 1-2\alpha.
\]
For the lower tail bound, we derive a bound for $\expect{Z^4}$ and invoke the Paley--Zygmund inequality. Specifically, we have
\begin{align*}
\expect{Z^4} &= \sum_{i_1, i_2, i_3,i_4}\delta_{i_1}\delta_{i_2}\delta_{i_3}\delta_{i_4}\expect{X_{i_1}X_{i_2}X_{i_3}X_{i_4}}
\\
&=m_{\{1,1,1,1\}}\Sigma_1 +  m_{\{1,1,1,2\}}\Sigma_{2,1}
+ m_{\{1,1,2,2\}}\Sigma_{2,2} + m_{\{1,1,2,3\}}\Sigma_3
+ m_{\{1,2,3,4\}}\Sigma_4,
\end{align*}
where we have abbreviated
\begin{align*}
\Sigma_1 &= \normfour{\delta}^4,
\\
\Sigma_{2, 1}&= 4\sum_{i_1< i_2}\left(\delta_{i_1}^3\delta_{i_2}
+ \delta_{i_2}^3\delta_{i_1}\right),
\\
\Sigma_{2, 2}&= 6\sum_{i_1< i_2} \delta_{i_1}^2\delta_{i_2}^2,
\\
\Sigma_{3}&= 12\sum_{i_1<i_2<i_3}
\left(\delta_{i_1}^2\delta_{i_2}\delta_{i_3}+\delta_{i_1}\delta_{i_2}^2\delta_{i_3}
+\delta_{i_1}\delta_{i_2}\delta_{i_3}^2\right),
\\
\Sigma_{4}&= 24\sum_{i_1<i_2<i_3<i_4}\delta_{i_1}\delta_{i_2}\delta_{i_3}\delta_{i_4}.
\end{align*}
The expressions for $\Sigma$'s above can be simplified further by using $\sum_{i\in[\ab]}\delta_i=0$. Observe now that
\[
\Sigma_1+\Sigma_{2,1}+\Sigma_{2,2} + \Sigma_3 +\Sigma_4 = \Big(\sum_{i\in [\ab]}\delta_i\Big)^4 =0.
\]
Also, for $\Sigma_{2,1}$ and $\Sigma_{2,2}$, we obtain
\begin{align*}
\Sigma_{2,1} &= 4\left(\sum_{i}\delta_{i}  \sum_{i}\delta_{i}^3 - \sum_{i}\delta_i^4\right) = -4\normfour{\delta}^4,
\\
\Sigma_{2,2} &= 3\left(\sum_{i}\delta_{i}^2  \sum_{i}\delta_{i}^2 - \sum_{i}\delta_i^4\right) =3\normtwo{\delta}^4 - 3\normfour{\delta}^4.
\end{align*}
Finally, the expressions for $\Sigma$'s can be seen to satisfy,
\[
\left(\sum_{i}\delta_{i}\right)^2 \sum_{i}\delta_{i}^2 = \sum_i \delta_i^4 + \frac 1 2 \Sigma_{2,1} + \frac 1 3 \Sigma_{2,2} + \frac 1 6 \Sigma_3,
\]
whereby
\begin{align*}
\Sigma_3 = - 6\Sigma_1 -3\Sigma_{2,1} - 2\Sigma_{2,2}
= 12\normfour{\delta}^4 - 6\normtwo{\delta}^4 
\end{align*}
Combining the relations above, we obtain
\begin{align*}
\expect{Z^4} &= \big(m_{\{1,1,1,1\}} -4m_{\{1,1,1,2\}} -3m_{\{1,1,2,2\}} +12 m_{\{1,1,2,3\}} - 6m_{\{1,2,3,4\}}\big) \normfour{\delta}^4 
\\
&\qquad+ \big(3m_{\{1,1,2,2\}} - 6m_{\{1,1,2,3\}}+3m_{\{1,2,3,4\}}\big)\normtwo{\delta}^4
\\
&\leq \big(m_{\{1,1,1,1\}}+ 3m_{\{1,1,2,2\}}+ 15 m_{\{1,1,2,3\}}\big)\normtwo{\delta}^4.
\end{align*}
Note that by symmetry 
\[
2(m_{\{1,1,2,2\}} - m_{\{1,1,2,3\}})= E[X_1^2(X_2-X_3)^2] \geq 0
\]
and by symmetry and the Cauchy--Schwarz inequality
\[
m_{\{1,1,1,1\}}=\expect{X_1^4}\geq \expect{X_1^2X_2^2} = m_{\{1,1,2,2\}}.
\]
Therefore, the previous inequality yields
\begin{align}\label{eq:paley:2}
\expect{Z^4}&\leq 19 \expect{X_1^4}\normtwo{\delta}^4\,.
\end{align}
We now take recourse to the Paley--Zygmund inequality, restated below:
\begin{theorem}[Paley--Zygmund (Refined version)]\label{theo:paley:zygmund}
    Suppose $U$ is a non-negative random variable with finite variance. Then, for every $\theta\in[0,1]$, 
    \begin{equation}
        \probaOf{ U > \theta\expect{U} } \geq \frac{(1-\theta)^2\expect{U}^2}{\var U + (1-\theta)^2\expect{U}^2}\,.
    \end{equation}
\end{theorem}
Applying this to $Z^2$ and substituting the bounds of~\cref{eq:paley:1,eq:paley:2} above, and setting
\[
\theta \eqdef 1-\sqrt{\frac{38\alpha}{1-2\alpha}}\frac{\sqrt{\expect{X_1^4}\normtwo{\delta}^4}}{(\expect{X_1^2}-\expect{X_1X_2})\normtwo{\delta}^2}
\leq 1-\sqrt{\frac{19\alpha}{1-2\alpha}}\frac{\sqrt{\expect{Z^4}}}{\expect{Z^2}}\,,
\]
we obtain
\[
\probaOf{Z^2 \geq \theta (\expect{X_1^2}-\expect{X_1X_2})\normtwo{\delta}^2}
\geq \frac{(1-\theta)^2(\expect{X_1^2}-\expect{X_1X_2})^2\normtwo{\delta}^4}{
19\expect{X_1^4}\normtwo{\delta}^4+ (1-\theta)^2(\expect{X_1^2}-\expect{X_1X_2})^2\normtwo{\delta}^4} = 2\alpha\,,
\]
which completes the proof.
\end{proof}

\section{Proof of~\cref{lemma:ub-rappor:variance}}\label{app:variance:rappor}
\begin{lemma}[Variance of the \Rappor-based estimator, restated]
  With $T$ defined as in~\eqref{eq:rappor:nx}, we have
    \[
        \variance{T} \leq 4\ab\ns^2 + 8\ns^3 \alpha_{R}^2 \normtwo{\p-\uniform}^2 = 4\ab\ns^2 + 8\ns\expect{T}\,.
    \]
\end{lemma}
\begin{proofof}{\cref{lemma:ub-rappor:variance}}
    Letting again $\lambda \eqdef \frac{\alpha_{R}}{\ab}+\beta_{R}$ and $\lambda_x\eqdef \frac{1}{\ns}\expect{N_x} = \alpha_{R}\p_x+\beta_{R}$ for $x\in[\ab]$, we have
    \[
          \variance{T} = \sum_{x,y\in[\ab]} \covariance{ f(N_x) }{ f(N_y) }\,,
    \]
    where $f\colon[0,\infty)\to\RR$ is given by $f(t) =
      (t-(\ns-1)\lambda)^2 - t + (\ns-1)\lambda^2$. 

The key difficulty in analysis arises from the fact that 
      \Rappor renders the multiplicities
      $N_x$'s dependent random variable. They are negatively associated, but since $f$
      is not monotone this does not imply that the cross covariance
      terms are non-positive. Thus, we need to take recourse to a more
      direct, elaborate treatment.

Fix $x\neq y$ in $[\ab]$. Expanding the covariance term, recalling $\expect{f(N_x)}  =
      \ns(\ns-1)(\lambda_x-\lambda)^2$ from the proof
      of~\cref{lemma:ub-rappor:expect}  and abbreviating $r\eqdef
      2(\ns-1)\lambda+1$ and $m\eqdef\ns-1$, we obtain after a few
      manipulations that 
    \begin{align*}
      &\covariance{ f(N_x) }{ f(N_y) }\\
      &= \expect{ (f(N_x)-\expect{f(N_x)})(f(N_y)-\expect{f(N_y)})}\\
      &= \expect{N_x^2N_y^2} - r\expect{N_x^2N_y+N_xN_y^2} 
      + m\ns\lambda_y(2\lambda-\lambda_y)\expect{N_x^2}
      + m\ns\lambda_x(2\lambda-\lambda_x)\expect{N_y^2} \\
      &\qquad + r^2\expect{N_xN_y} - rm\ns( \lambda_y(2\lambda-\lambda_y)\expect{N_x} + \lambda_x(2\lambda-\lambda_x)\expect{N_y} )\\
      &\qquad + m^2\ns^2 \lambda_x\lambda_y(2\lambda-\lambda_x)(2\lambda-\lambda_y)\\
      &= \expect{N_x^2N_y^2} - r\expect{N_x^2N_y+N_xN_y^2} 
      - 2m^2\ns\lambda\lambda_y(2\lambda-\lambda_y)\expect{N_x}
      + m^2\lambda_y(2\lambda-\lambda_y)\expect{N_x}^2\\
      &\qquad-2m^2\ns\lambda\lambda_x(2\lambda-\lambda_x)\expect{N_y}
      + m^2\lambda_x(2\lambda-\lambda_x)\expect{N_y}^2 \\
      &\qquad + r^2\expect{N_xN_y}  
      + m^2\ns^2 \lambda_x\lambda_y(2\lambda-\lambda_x)(2\lambda-\lambda_y).
    \end{align*}
 Substituting $r-1=2m\lambda$ and $\expect{N_y^2} =
 \expect{N_y}+\frac{m}{\ns}\expect{N_y}^2$ in the previous identity,
 and similarly for $N_x$, we get
    \begin{align}\label{eq:covariance}
& \covariance{ f(N_x) }{ f(N_y) } \nonumber \\ &= \expect{N_x^2N_y^2}
      - (2m\lambda+1)\expect{N_x^2N_y+N_xN_y^2} -
      2m^2\lambda(2\lambda-\lambda_y)\expect{N_x}\expect{N_y}
      \notag\\ &\qquad-2m^2\lambda(2\lambda-\lambda_x)\expect{N_x}\expect{N_y}
      + m^2\lambda_x(2\lambda-\lambda_x)\expect{N_y}^2 +
      m^2\lambda_y(2\lambda-\lambda_y)\expect{N_x}^2\notag\\ &\qquad +
      (2m\lambda+1)^2\expect{N_xN_y} +
      m^2\expect{N_x}\expect{N_y}(2\lambda-\lambda_x)(2\lambda-\lambda_y).
\end{align}
    We proceed by evaluating the expressions for $\expect{N_xN_y}$,
    $\expect{N_x^2N_y}$, $\expect{N_xN_y^2}$, and
    $\expect{N_x^2N_y^2}$ separately.

Specifically, by~\cref{fact:rappor:statistics}, we get
    \begin{align}
      \expect{N_xN_y} &= \sum_{1\leq i,j\leq \ns}  \probaOf{\bb_{ix}
        =1, \bb_{jy}=1}
\nonumber
\\
     & = \sum_{i=1}^\ns ( \lambda_x\lambda_y - \alpha_{R}^2\p(x)\p(y)
) + \sum_{i\neq j} \lambda_x\lambda_y 
\nonumber
\\
&= \ns^2\lambda_x\lambda_y - \ns\alpha_{R}^2\p(x)\p(y).
\label{eq:covariance:11}
    \end{align}
    
Turning to $\expect{N_x^2N_y}$, we get
    \begin{align*}
        \expect{N_x^2 N_y}
        &= \sum_{1\leq i,j,\ell\leq \ns} \probaOf{ \bb_{ix} = 1, \bb_{jx} = 1, \bb_{\ell y} = 1 } \\
        &= \ns\probaOf{ \bb_{ix} = 1, \bb_{i y} = 1 } + 6\binom{\ns}{3} \lambda_x^2 \lambda_y + 2\binom{\ns}{2}\left( \lambda_x\lambda_y+2\lambda_x(\lambda_x\lambda_y-\alpha_{R}^2\p(x)\p(y)) \right) \\
        &= \ns \lambda_x\lambda_y - \ns\alpha_{R}^2\p(x)\p(y) + \ns(\ns-1)(\ns-2)\lambda_x^2 \lambda_y  \\
        &\qquad+ \ns(\ns-1) \lambda_x\lambda_y+ 2 \ns(\ns-1)\lambda_x^2\lambda_y - 2 \ns(\ns-1)\alpha_{R}^2\lambda_x\p(x)\p(y),
    \end{align*}
which yields
    \begin{equation}\label{eq:covariance:21}
        \expect{N_x^2 N_y} = \ns^2 \lambda_x\lambda_y - (2m\lambda_x+1)\ns\alpha_{R}^2\p(x)\p(y) + m\ns^2\lambda_x^2 \lambda_y.
    \end{equation}
    
For the last term, note that
    \begin{align*}
        \expect{N_x^2 N_y^2} 
        &= \sum_{1\leq i,j,i',j'\leq \ns} \probaOf{ \bb_{ix} = 1, \bb_{jx} = 1, \bb_{i' y} = 1, \bb_{j' y} = 1 } \\
        &= \ns\Paren{\lambda_x\lambda_y-\alpha_{R}^2\p(x)\p(y)} 
        + \binom{\ns}{2}\Paren{2\lambda_x\lambda_y + 4\lambda_x\Paren{\lambda_x\lambda_y-\alpha_{R}^2\p(x)\p(y)} \right.\\&\qquad\left.+ 4\lambda_y\Paren{\lambda_x\lambda_y-\alpha_{R}^2\p(x)\p(y)} + 4\Paren{\lambda_x\lambda_y-\alpha_{R}^2\p(x)\p(y)}^2 } \\
        &\qquad+ \binom{\ns}{3}\Paren{ 6\lambda_x^2\lambda_y+6\lambda_x\lambda_y^2 + 24\lambda_x\lambda_y\Paren{\lambda_x\lambda_y-\alpha_{R}^2\p(x)\p(y)}  } \\
        &\qquad + 24\binom{\ns}{4} \lambda_x^2\lambda_y^2
    \end{align*}
    where the second identity follows from counting the different
    possibilities for the values taken by $i,i',j,j'$; 
    we divide into cases based on the number of different values taken
    and apply~\cref{fact:rappor:statistics} for each subcase. Note
    that the total number of terms is
    $\ns+14\binom{\ns}{2}+36\binom{\ns}{3}+24\binom{\ns}{4}=\ns^4$. 

This after a tedious simplification leads to
    \begin{align}\label{eq:covariance:22}
        \expect{N_x^2 N_y^2}
        &= \ns\Paren{\lambda_x\lambda_y-\alpha_{R}^2\p(x)\p(y)} 
        + m\ns\Paren{\lambda_x\lambda_y + 2\lambda_x\Paren{\lambda_x\lambda_y-\alpha_{R}^2\p(x)\p(y)} \right.\notag\\&\qquad\left.+ 2\lambda_y\Paren{\lambda_x\lambda_y-\alpha_{R}^2\p(x)\p(y)} + 2\Paren{\lambda_x\lambda_y-\alpha_{R}^2\p(x)\p(y)}^2 } \notag\\
        &\qquad+ m\ns(m-1)\Paren{ \lambda_x^2\lambda_y+\lambda_x\lambda_y^2 + 4\lambda_x\lambda_y\Paren{\lambda_x\lambda_y-\alpha_{R}^2\p(x)\p(y)}  } \notag\\
        &\qquad + m\ns(m-1)(m-2) \lambda_x^2\lambda_y^2 \notag\\
       &= m^2\ns^2 \lambda_x^2 \lambda_y^2 + m\ns^2 (\lambda_x^2 \lambda_y + \lambda_x \lambda_y^2) + \ns^2 \lambda_x \lambda_y 
       - 4\alpha_{R}^2m^2\ns \p(x) \p(y) \lambda_x \lambda_y \notag\\
&\qquad- 2 \alpha_{R}^2 m\ns \Paren{ \p(x) \p(y) \lambda_x + \p(x) \p(y) \lambda_y } + 2\alpha_{R}^4 m\ns \p(x)^2 \p(y)^2 -\alpha_{R}^2 \ns \p(x) \p(y)\,, 
    \end{align}

Upon
combining~\cref{eq:covariance:11,eq:covariance:21,eq:covariance:22}
with~\eqref{eq:covariance} and further simplifying the expressions, we get
    \begin{align*}
    \covariance{ f(N_x) }{ f(N_y) } &= 2 \alpha_{R}^4 \ns(\ns-1) \Paren{ \p(x)^2 \p(y)^2 - 2(\ns-1) \p(x) \p(y) \left(\p(x)-\tfrac{1}{\ab}\right)\left(\p(y)-\tfrac{1}{\ab}\right) }\,.
    \end{align*}
    We proceed by summing both sides over pairs of distinct
    $x,y\in[\ab]$ to obtain
    \begin{align*}
        &\sum_{x\neq y}\covariance{ f(N_x) }{ f(N_y) }
\\
        &= 2\alpha_{R}^4\ns(\ns-1)\sum_{x\neq y} \p(x)^2\p(y)^2 
         - 4\alpha_{R}^4\ns(\ns-1)^2 \sum_{x\neq y} \p(x)\p(y)\left(\p(x)-\tfrac{1}{\ab}\right)\left(\p(y)-\tfrac{1}{\ab}\right) \\
        &= 2\alpha_{R}^4\ns(\ns-1)\Paren{ \normtwo{\p}^4 - \normfour{\p}^4 }
         - 4\alpha_{R}^4\ns(\ns-1)^2 \Paren{ \normtwo{\p-\uniform}^4 - \sum_{x\in[\ab]} \p(x)^2\left(\p(x)-\tfrac{1}{\ab}\right)^2 }\\
         &\leq 2\alpha_{R}^4\ns^2+ 4\alpha_{R}^4\ns(\ns-1)^2 \normtwo{\p-\uniform}^2\\
         &\leq 2\ns^2 + 4\alpha_{R}^2\ns(\ns-1)^2\normtwo{\p-\uniform}^2\,,
    \end{align*}
 where we have   used $\sum_{x\in[\ab]}
 \p(x)\left(\p(x)-\frac{1}{\ab}\right) = \normtwo{\p}^2-1/\ab =
 \normtwo{\p-\uniform}^2$ and bounded the non-negative terms.

This completes our bound for the cross-variance terms. Turning now to
the variances, we note that
    \begin{align*}
          \variance{ f(N_x) } 
          &= 2 \ns(\ns-1) \lambda_x^2(1-\lambda_x)^2 + 4\ns(\ns-1)^2\lambda_x(1-\lambda_x) (\lambda-\lambda_x)^2
          \\
&\leq \frac{1}{8}\ns^2 + \alpha_{R}^2\ns(\ns-1)^2\left(\p(x)-\tfrac{1}{\ab}\right)^2,
    \end{align*}
where the previous inequality holds since $\lambda_x\in[0,1]$. It
follows that
    \begin{align*}
        \sum_{x\in[\ab]}\variance{ f(N_x) }
        &\leq \frac{1}{8}\ns^2\ab + \alpha_{R}^2\ns(\ns-1)^2\normtwo{\p-\uniform}^2\,.
    \end{align*}
Putting everything together, we conclude
    \[
        \variance{T} = \sum_{x\in[\ab]}\variance{ f(N_x) } + \sum_{x\neq y}\covariance{ f(N_x) }{ f(N_y) } \leq \frac{17}{8}\ns^2 + 5\alpha_{R}^2\ns(\ns-1)^2\normtwo{\p-\uniform}^2\,,
    \]
    proving the lemma.
\end{proofof}

\section{Proof of~\cref{theorem:random:product:subsets}}\label{app:concentration:bivariate}

\begin{theorem}[Joint probability perturbation concentration, restated]
Consider a matrix $\delta\in\RR^{\ab\times \ab}$ such that, for every $i_0,j_0\in[\ab]$, $\sum_{j\in [\ab]}\delta_{i_0,j}=\sum_{i\in [\ab]}\delta_{i,j_0} = 0$. Let random variables $X=(X_1, \dots, X_\ab)$ and $Y=(Y_1, \dots, Y_\ab)$ be independent and uniformly distributed over $\ab$-length binary sequences of weight $\ab/2$.
Define $Z = \sum_{(i,j)\in [\ab]\times[\ab]}\delta_{ij}X_iY_j$. 
Then,  there exist constants $c_1,c_2,\rho>0$ such that
\[
\probaOf{\frac{Z^2}{ \norm{\delta}_F^2}\in [c_1,c_2]}\geq \rho.
\]
\end{theorem}
\begin{proof}
  The proof is similar to that of~\cref{theorem:random:subset}, with further technicalities difficulties arising since the random variables $(X_i,Y_j)$ assigned as weights to $\delta_{ij}$, $1\leq i,j\leq \ab$, are not $4$-symmetric (as a pair). However, $\delta$ has an additional structure since both its rows and columns sum to zero; we complete the proof by exploiting this property and the fact that $X_i$'s and $Y_j$'s are individually $4$-symmetric. As before, we use the notation $m_S \eqdef \expect{\prod_{i\in S} X_i}$, and moreover, let $m'_S \eqdef \expect{\prod_{i\in S} Y_i}$.
  
   First, observe that by independence of $X_i$'s and $Y_j$'s and since $\sum_{(i,j)\in [\ab]\times[\ab]}\delta_{ij} = 0$, we have 
  \begin{equation*}
      \expect{Z} = \sum_{(i,j)\in [\ab]\times[\ab]}\delta_{ij} \expect{X_i}\expect{Y_j} = m_{\{1\}}m'_{\{1\}}\sum_{(i,j)\in [\ab]\times[\ab]}\delta_{ij} = 0\,.
  \end{equation*}
Furthermore, 
  \begin{align*}
\variance{Z} = \expect{Z^2} 
&=  \sum_{(i_1,j_1)\in[\ab]\times[\ab]}\sum_{(i_2,j_2)\in[\ab]\times[\ab]}\delta_{i_1j_1}\delta_{i_2j_2} \expect{X_{i_1}X_{i_2}}\expect{Y_{j_1}Y_{j_2}}  \\
&=  \sum_{i_1,i_2} \expect{X_{i_1}X_{i_2}} \sum_{j_1,j_2}\delta_{i_1j_1}\delta_{i_2j_2}\expect{Y_{j_1}Y_{j_2}}\\
&=  \sum_{i_1,i_2} \expect{X_{i_1}X_{i_2}} \Paren{ m'_{\{1,1\}}\sum_{j_1}\delta_{i_1j_1}\delta_{i_2j_1} + m'_{\{1,2\}} \sum_{j_1\neq j_2}\delta_{i_1j_1}\delta_{i_2j_2} }\\
&=  \sum_{i_1,i_2} \expect{X_{i_1}X_{i_2}} \Bigg(
(m'_{\{1,1\}}-m'_{\{1,2\}})\sum_{j}\delta_{i_1j}\delta_{i_2j} +
m'_{\{1,2\}} \Big(\sum_{j}\delta_{i_1j} \Big)^2 \Bigg)\\
&= (m'_{\{1,1\}}-m'_{\{1,2\}})\sum_{i_1,i_2}  \sum_{j}\delta_{i_1j}\delta_{i_2j} \expect{X_{i_1}X_{i_2}}\, ,
  \end{align*}
where the previous identity uses $\sum_{j}\delta_{i_1j} = 0$. Repeating the same manipulations with the outer sum, we get
  \begin{align*}
\variance{Z} = \expect{Z^2}  &= (m'_{\{1,1\}}-m'_{\{1,2\}})\sum_{j} \Bigg( m_{\{1,1\}}\sum_{i} \delta_{ij}^2 + m_{\{1,2\}}\sum_{i_1\neq i_2} \delta_{i_1j}\delta_{i_2j}  \Bigg) \notag\\
&= (m'_{\{1,1\}}-m'_{\{1,2\}})\sum_{j} \Bigg( (m_{\{1,1\}}-m_{\{1,2\}})\sum_{i} \delta_{ij}^2 + m_{\{1,2\}}\Big(\sum_{i} \delta_{i,j}\Big)^2  \Bigg) \notag\\
&= (m_{\{1,1\}}-m_{\{1,2\}})(m'_{\{1,1\}}-m'_{\{1,2\}})\norm{\delta}_F^2
\,. 
  \end{align*}
Up to this point, our calculations are valid for any independent choice of $4$-symmetric $(X_1,\dots,X_\ab)$ and $(Y_1, \dots, Y_\ab)$. For our specific choice, using calculations from the proof of~\cref{corollary:random:subset}, we obtain the following:
\[
\frac{1}{16} \norm{\delta}_F^2 \leq  \variance{Z} = \frac 1 4 \expect{(X_1-X_2)^2}\expect{(Y_1-Y_2)^2} \norm{\delta}_F^2 \leq \frac{1}{4}\norm{\delta}_F^2\,.
\]
It remains to bound the fourth moment of $Z$; for simplicity, we provide this proof only for our specific choice of random variables. We
have
\begin{align}
\expect{Z^4} &=
\sum_{(i_1,j_1)}\sum_{(i_2,j_2)}\sum_{(i_3,j_3)}\sum_{(i_4,j_4)}\delta_{i_1j_1}\delta_{i_2j_2}\delta_{i_3j_3}\delta_{i_4j_4}
\expect{X_{i_1}X_{i_2}X_{i_3}X_{i_4}}\expect{Y_{j_1}Y_{j_2}Y_{j_3}Y_{j_4}}
\notag\\ &=
\sum_{i_1,i_2,i_3,i_4} \expect{X_{i_1}X_{i_2}X_{i_3}X_{i_4}}
\sum_{j_1,j_2,j_3,j_4}
\expect{Y_{j_1}Y_{j_2}Y_{j_3}Y_{j_4}}
\delta_{i_1j_1}\delta_{i_2j_2}\delta_{i_3j_3}\delta_{i_4j_4}
\nonumber\,.
  \end{align}
Consider the inner summation for an arbitrary fixed choice of $i_1,
i_2, i_3, i_4$. We show first that the inner summation can be
expressed as
\begin{align}
\alpha_\ab\sum_{j}\delta_{i_1j}\delta_{i_2j}\delta_{i_3j}\delta_{i_4j}
+ \beta_k\sum_{j}\delta_{i_1j}\delta_{i_2j}\cdot
\sum_j\delta_{i_3j}\delta_{i_4j},
\label{e:remaining:term}
\end{align}
for appropriate coefficients $\alpha_\ab$ and $\beta_\ab$. To that end, we introduce
the notations: For $\ell\leq 4$,  denote
\begin{align*}
p_\ell(j_1, \dots, j_{\ell}) &\eqdef \probaOf{Y_{j_1}=1, \dots, Y_{j_\ell}=1}
\\
p_\ell(j_1, \dots, j_{\ell-1}) &\eqdef \probaOf{Y_{j_\ell}=1\mid Y_{j_1}=1, \dots, Y_{j_{\ell-1}}=1}, \quad j_\ell\notin\{j_1, \dots, j_{\ell-1}\},
\\
p_\ell(j_1, \dots, j_{m}) &\eqdef \probaOf{Y_{j_\ell}=1\mid Y_{j_1}=1, \dots, Y_{j_{m}}=1, Y_{j_{m+1}} =1, \dots, 
Y_{j_{\ell-1}} =1 },
\\
&\qquad\qquad\qquad j_r \notin \{j_1, \dots, j_{m}\},\,\, j_r \text{ distinct for } m+1< r\leq \ell,
\end{align*}
and the corresponding notation $q_\ell(j_1, \dots, j_{m}) \eqdef 1- p_\ell(j_1, \dots, j_m)$. With this in hand, we can write
\begin{align*}
\sum_{j_1,j_2,j_3,j_4}&\delta_{i_1j_1}\delta_{i_2j_2}\delta_{i_3j_3}\delta_{i_4j_4}\expect{Y_{j_1}Y_{j_2}Y_{j_3}Y_{j_4}}
\\
&=\sum_{j_1,j_2,j_3}\delta_{i_1j_1}\delta_{i_2j_2}\delta_{i_3j_3}\left(\sum_{j_4\in
\{j_1, j_2, j_3\}}\delta_{i_4j_4}
\expect{Y_{j_1}Y_{j_2}Y_{j_3}} + \sum_{j_4\notin
\{j_1, j_2, j_3\}}\delta_{i_4j_4}
\expect{Y_{j_1}Y_{j_2}Y_{j_3}Y_{j_4}}\right) 
\\
&=\sum_{j_1,j_2,j_3}p_3(j_1,
  j_2, j_3)\delta_{i_1j_1}\delta_{i_2j_2}\delta_{i_3j_3}\left(\sum_{j_4\in
\{j_1, j_2, j_3\}}\delta_{i_4j_4}
 + p_4(j_1, j_2, j_3)\sum_{j_4\notin
\{j_1, j_2, j_3\}}\delta_{i_4j_4}\right) 
\\
&=\sum_{j_1,j_2,j_3}p_3(j_1,
  j_2, j_3)q_4(j_1, j_2, j_3)\delta_{i_1j_1}\delta_{i_2j_2}\delta_{i_3j_3}\sum_{j_4\in
\{j_1, j_2, j_3\}}\delta_{i_4j_4},
\end{align*}
where we have used $\sum_{j}\delta_{i_4, j}=0$ in the previous
identity. Next, note that each term in the sum above has, by symmetry, the form
\begin{align*}
&\sum_{j_1,j_2,j_3}p_3(j_1,j_2, j_3)q_4(j_1, j_2,j_3)\delta_{i_1j_1}\delta_{i_2j_2}\delta_{i_3j_3}\delta_{i_4j_1}
\\
&=\sum_{j_1,j_2}p_2(j_1,
  j_2) \Big(
q_3(j_1, j_2)\sum_{j_3\in \{j_1, j_2\}}\delta_{i_1j_1}\delta_{i_2j_2}\delta_{i_3j_3}\delta_{i_4j_1}
+ p_3(j_1, j_2)q_4(j_1, j_2)\sum_{j_3\notin \{j_1, j_2\}}\delta_{i_1j_1}\delta_{i_2j_2}\delta_{i_3j_3}\delta_{i_4j_1}\Big)
\\
&=\sum_{j_1,j_2}p_2(j_1,
  j_2) \Big(1- p_3(j_1, j_2)\big(1+q_4(j_1, j_2)\big)\Big)
\sum_{j_3\in \{j_1, j_2\}}\delta_{i_1j_1}\delta_{i_2j_2}\delta_{i_3j_3}\delta_{i_4j_1},
\end{align*}
where we have once again used $\sum_{j_3}\delta_{i_3, j_3}=0$ to obtain the last line from the second-to-last. This
leaves us with terms of the form
\begin{align*}
\sum_{j_1,j_2}p_2(j_1,j_2)\Big(1- p_3(j_1, j_2)\big(1+q_4(j_1, j_2)\big)\Big)
\left( \delta_{i_1j_1}\delta_{i_2j_2}\delta_{i_3j_1}\delta_{i_4j_1}
+ \delta_{i_1j_1}\delta_{i_2j_2}\delta_{i_3j_2}\delta_{i_4j_1}\right).
\end{align*}
Finally, splitting the summation above into terms with $j_1=j_2$ and
$j_1\neq j_2$, we end up with terms of the form
\eqref{e:remaining:term}. 
Note that the resulting coefficients $\alpha_\ab,\beta_\ab$ entail terms dependent on $\ab$ which can be handled and bounded (crucially, independently of $\ab$) in the manner of proof of~\cref{corollary:random:subset}. 

To complete the proof, we handle each term in~\eqref{e:remaining:term}
separately. For the first, we obtain
\begin{align}
\sum_{j}\sum_{i_1,i_2,i_3,
  i_4}\delta_{i_1j}\delta_{i_2j}\delta_{i_3j}\delta_{i_4j}
\expect{X_1X_2X_3X_4} \leq 19\expect{X_1^4}\sum_j\Big(
\sum_{i}\delta_{ij}^2\Big)^2
\leq 19\expect{X_1^4}\norm{\delta}_F^2\,,
\nonumber
\end{align}
where the inequality uses~\eqref{eq:paley:2}. For the second term, note that
\begin{align*}
\sum_{i_1,i_2,i_3,i_4}X_{i_1}X_{i_2}X_{i_3}X_{i_4}
\sum_{j}\delta_{i_1j}\delta_{i_2j}\cdot
\sum_{j^\prime}\delta_{i_3j^\prime}\delta_{i_4j^\prime}
&= \Big(\sum_{j}\Big(\sum_{i}X_i\delta_{ij}\Big)^2\Big)^2,
\end{align*}
whereby the expected value of the left-side can be bounded by
\begin{align*}
\expect{\left(\sum_{j}\left(\sum_{i}X_i\delta_{ij}\right)^2\right)^2}
&= \sum_{j, j^\prime}\expect{\left(\sum_{i}X_i\delta_{ij}\right)^2\left(\sum_{i}X_i\delta_{ij^\prime}\right)^2}
\\
&\leq \sum_{j, j^\prime}\sqrt{\expect{\left(\sum_{i}X_i\delta_{ij}\right)^4}\expect{
\left(\sum_{i}X_i\delta_{ij^\prime}\right)^4}}
\\
&\leq 19 \expect{X_1^4} \sum_{j, j^\prime} \normtwo{\delta_{\cdot, j}}^2\normtwo{\delta_{\cdot, j^\prime}}^2,
\end{align*}
where the final inequality uses \eqref{eq:paley:2}. The sum on the right-side can be seen to simplify as
\begin{align*}
\sum_{j, j^\prime} \normtwo{\delta_{\cdot, j}}^2\normtwo{\delta_{\cdot, j^\prime}}^2
= \sum_{j, j^\prime}\sum_{i}\delta_{i,j}^2 \sum_{i}\delta_{i,j^\prime}^2
= \left(\sum_{ij}\delta_{i,j}^2\right)^2 
=\norm{\delta}_F^4.
\end{align*}
Therefore, on combining the observations above, we obtain
\[
\expect{Z^4}\leq c^\prime  \norm{\delta}_F^4,
\]
for an appropriate absolute constant $c'$. The proof is completed in the manner of that of~\cref{theorem:random:subset} using the Paley--Zygmund inequality.
\end{proof}

\end{document}